\documentclass[journal,12pt,one column]{IEEEtran}
\usepackage{amsfonts}
\usepackage{amssymb}
\usepackage{amsmath}
\usepackage[dvips]{graphicx}

\input{tcilatex}
\newtheorem{theo}{Theorem}
\newtheorem{lemm}{Lemma}
\newtheorem{coro}{Corollary}

\newtheorem{rema}{Remark}

\begin{document}

\title{MIMO Two-Way Relaying: A Space-Division Approach}
\author{ \thanks{
This manuscript is submitted to IEEE Transactions on Information Theory} \authorblockA {Xiaojun Yuan, \emph{Member}, \emph{IEEE}, Tao Yang, \emph{Member}, \emph{IEEE}, and Iain B. Collings, \emph{Senior Member}, \emph{IEEE},
} }
\maketitle

\begin{abstract}
We propose a novel space-division based network-coding scheme for
multiple-input multiple-output (MIMO) two-way relay channels (TWRCs), in
which two multi-antenna users exchange information via a multi-antenna
relay. In the proposed scheme, the overall signal space at the relay is
divided into two subspaces. In one subspace, the spatial streams of the two
users have nearly orthogonal directions, and are completely decoded at the
relay. In the other subspace, the signal directions of the two users are
nearly parallel, and linear functions of the spatial streams are computed at
the relay, following the principle of physical-layer network coding (PNC).
Based on the recovered messages and message-functions, the relay generates
and forwards network-coded messages to the two users. We show that, at high signal-to-noise ratio
(SNR), the proposed scheme achieves the asymptotic sum-rate capacity of MIMO TWRCs within $\frac{1}{2}%
\log(5/4)\approx 0.161$ bits per user-antenna for any
antenna configuration and channel realization. We perform large-system analysis to derive the average sum-rate of the proposed scheme
over Rayleigh-fading MIMO TWRCs. We show that the average asymptotic sum-rate gap to the capacity upper bound is at most
0.053 bits per relay-antenna. It is demonstrated that the proposed scheme
significantly outperforms the existing schemes.
\end{abstract}

\section{Introduction}

During the past decade, tremendous progress has been made in the field of
network coding \cite{RAhlswededIT00}. In \cite{ZhangMobicom06}-\cite{Petar07}%
, the concept of physical-layer network coding (PNC) was introduced and
applied to wireless networks. The simplest model for wireless PNC is a
two-way relay channel (TWRC), in which two users $A$ and $B$ exchange
information via an intermediate relay. Compared with conventional schemes,
PNC allows the relay to deliver linear functions of the users'
messages, which can potentially double the network throughput. It has been
shown that the PNC scheme can achieve the capacity of a Gaussian TWRC within
1/2 bit per user \cite{NamIT10}\cite{WilsonIT10}, and its gap to the
capacity vanishes at high signal-to-noise ratio (SNR).

Recently, efficient communications over MIMO TWRCs have attracted much
research interest, where the two users and the relay are equipped with
multiple antennas. Most work on MIMO TWRCs focuses on classical relaying
strategies borrowed from one-way relay channels, such as amplify-and-forward
(AF) \cite{KattiSIGCOMM07}-\cite{XuICCASP10}, compress-and-forward \cite%
{AleksicIT09}\cite{LimIT11}, and decode-and-forward (DF) \cite{KramerIT05}-%
\cite{GunduzAsilomar08}. These strategies generally perform well away from
the channel capacity due to noise amplification and multiplexing loss \cite%
{GunduzAsilomar08}. Recently, several relaying schemes have been proposed to
support PNC in MIMO TWRCs \cite{HJYangIT11}-\cite{NazerIT11}. The basic idea
is to jointly decompose the channel matrices of the two users to create
multiple scalar channels, over which multiple PNC streams are transmitted.
Let $n_{A}$, $n_{B}$, and $n_{R}$ denote the numbers of antennas of user $A$%
, user $B$, and the relay, respectively. For configurations with $%
n_{A},n_{B}\geq n_{R}$, a generalized singular-value-decomposition (GSVD)
scheme was shown to achieve the asymptotic capacity of MIMO TWRCs at high
SNR \cite{HJYangIT11}. For configurations with $n_{A},n_{B}<$ $n_{R}$, all
existing schemes may perform quite far away from the capacity. Such
configurations, however, are of most practical importance. For example, due
to the constrained physical sizes of the user terminals, it is usually
convenient to implement more antennas at the relay station than at the user
ends, as suggested in the standards of next generation networks \cite{3GPP}%
\cite{WiMAX}.

In this paper, we propose a new space-division based PNC scheme for MIMO
TWRCs. Specifically, we first establish a novel joint channel decomposition,
which characterizes the mutual orthogonality of the channel directions of
the two users seen at the relay. Based on this decomposition, the overall
signal space is divided into two orthogonal subspaces. In one subspace, the
channel directions of one user are orthogonal (or close to orthogonal) to
those of the other user. In this subspace, the spatial streams of the two
users are completely decoded. In the other subspace, the channel directions
of the two users are parallel or close to parallel. In this subspace, linear
functions of the corresponding spatial streams are computed, without
completely decoding the individual spatial streams. These linear
functions of the spatial streams are referred to as \textit{network-coded
messages}. The messages and the network-coded messages, respectively
generated from the two subspaces, are jointly encoded at the relay, and then
forwarded to the two users. Afterwards, the two users recover their desired
messages.

We derive the achievable rates of the proposed space-division based PNC
scheme for MIMO TWRCs. We analytically show that, in the high SNR regime,
the proposed scheme can achieve the sum capacity of the MIMO TWRC within $%
\min \{n_{A},n_{B}\}\log (5/4)$ bits, or $\frac{1}{2}\log (5/4)\approx 0.161$
bit/user-antenna, for any antenna setup and channel realization. This gap is
much smaller than (as low as 10\% of) the gap for the existing best scheme.
We also perform large-system analysis to derive the average achievable
sum-rate of the proposed scheme in Rayleigh fading MIMO TWRCs. We show that,
in the high SNR regime, the average gap between our scheme and the
sum-capacity upper bound is greatest when the antenna configuration is $%
n_{A}=n_{B}=\frac{1}{2}n_{R}$, with the gap being only $0.053$
bit/relay-antenna. For all other configurations, the proposed scheme perform
even closer to the capacity upper bound. Particularly, as the ratio $n_{A}/n_{R}$ (or $%
n_{B}/n_{R}$) tends to 0 or 1, the gap to the capacity upper bound vanishes.
All these analytical results agree well with the simulation. Numerical results demonstrate that the proposed scheme significantly
outperforms the existing schemes in the literature across the full range of
SNRs.

\section{Preliminaries}

\subsection{Notation}

The following notations are used throughout this paper. We use lowercase
regular letters for scalars, lowercase bold letters for vectors, and
uppercase bold letters for matrices. The superscripts $(\cdot)^T$ and $(\cdot)^\dag$ denote transpose and Hermitian transpose, respectively. $\|\cdot\|$ and $\|\cdot\|_F$ represent the Euclidian norm of a vector and the Frobenius norm of a matrix, respectively.  $\mathcal{C}(\mathbf{X})$ represents
the columnspace of a matrix $\mathbf{X}$. $\mathcal{%
\mathbb{R}
}^{n\times m}$ and $\mathcal{%
\mathbb{C}
}^{n\times m}$ denote the $n$-by-$m$ dimensional real space and complex
space, respectively. The operation $\log (\cdot )$ denotes the logarithm
with base 2, and $|\cdot |$ the determinant. $I(i)$ is the indicator
function with $I(i)=1$ for $i=1$ and $I(i)=0$ for $i\neq 1$; $[\cdot ]^{+}$
represents $\max \{\cdot ,0\}$; sign$(x)$ represents the sign of $x$; $%
\mathcal{N}_{c}(\mu ,\sigma ^{2})$ denotes the circularly symmetric complex
Gaussian distribution with mean $\mu $ and variable $\sigma ^{2}$.

\subsection{System Model}

In this paper, we consider a discrete memoryless MIMO TWRC in which users $A$
and $B$ exchange information via a relay, as illustrated in Fig. \ref%
{Fig_Config_MIMOTWRC1}. User $m$ is equipped with $n_{m}$ antennas, $m\in
\left\{ A,B\right\} $, and the relay with $n_{R}$ antennas. We assume that
there is no direct link between the two users. The channel is assumed to be
flat-fading and quasi-static, i.e., the channel coefficients remain
unchanged during each round of information exchange. The channel matrix from
user $m$ to the relay is denoted by $\mathbf{H}_{mR}\in \mathcal{%
\mathbb{C}
}^{n_{R}\times n_{m}}$, and that from the relay to user $m$ is
denoted by $\mathbf{H}_{Rm}\in \mathcal{%
\mathbb{C}
}^{n_{m}\times n_{R}},m\in \left\{ A,B\right\} $. We further assume that the
channel matrices are always of either full column rank or full row rank,
whichever is smaller, and are globally known by both users as well as by the
relay.

The system operates in a half-duplex mode. Two time-slots are employed for
each round of information exchange. Following the convention in \cite%
{HJYangIT11}-\cite{KhinaISIT11}, we assume that the two time-slots have same
duration. The extension of our results to the case of unequal duration is
straightforward.

In the first time-slot (referred to as \emph{uplink phase}), the two users
transmit to the relay simultaneously and the relay remains silent. The
transmit signal matrix at user $m$ is denoted by $\mathbf{X}_{m}\in \mathcal{%
\mathbb{C}
}^{n_{m}\times T},$ $m\in \left\{ A,B\right\} $, where $T$ is the number of
channel uses in one time-slot. Each column of $\mathbf{X}_{m}$ denotes the
signal vector transmitted by the $n_{m}$ antennas in one channel use. The
average power at each user is constrained as $\frac{1}{T}E\left[ \left\Vert
\mathbf{X}_{m}\right\Vert _{F}^{2}\right] \leq P_{m},$ $m\in \left\{
A,B\right\} $. The received signal at the relay is denoted by $\mathbf{Y}%
_{R}\in \mathcal{%
\mathbb{C}
}^{n_{R}\times T}$ with%
\begin{equation}
\mathbf{Y}_{R}=\mathbf{H}_{AR}\mathbf{X}_{A}+\mathbf{H}_{BR}\mathbf{X}_{B}+%
\mathbf{Z}_{R},  \label{YR}
\end{equation}%
where $\mathbf{Z}_{R}\in \mathcal{%
\mathbb{C}
}^{n_{R}\times T}$ denotes the additive white Gaussian noise (AWGN) at the
relay. We assume that the elements of $\mathbf{Z}_{R}$ are independent and
identically drawn from $\mathcal{N}_{c}(0,N_{0})$. Upon receiving $\mathbf{Y}%
_{R}$, the relay generates a signal matrix $\mathbf{X}_{R}\in \mathcal{%
\mathbb{C}
}^{n_{R}\times T}$.

In the second time-slot (referred to as \textit{downlink phase}), $\mathbf{X}%
_{R}$ is broadcast to the two users. The average power at the relay is
constrained as $\frac{1}{T}E\left[ \left\Vert \mathbf{X}_{R}\right\Vert
_{F}^{2}\right] \leq P_{R}.$ The signal matrix received by user $m$ is
denoted by $\mathbf{Y}_{m}\in \mathcal{%
\mathbb{C}
}^{n_{m}\times T},m\in \left\{ A,B\right\} $, with
\begin{equation}
\mathbf{Y}_{m}=\mathbf{H}_{Rm}\mathbf{X}_{R}+\mathbf{Z}_{m},m\in \left\{
A,B\right\} ,
\end{equation}%
where $\mathbf{Z}_{m}$ denotes the AWGN matrix at user $m$, with the
elements independently drawn from $\mathcal{N}_{c}(0,N_{0})$.

\subsection{Definition of Achievable Rates}

For the system model described above, denote the message of user $m$ by $%
W_{m}\in \{1,2,...,2^{2TR_{m}}\}$. The cardinality of $W_{m}$ is given by $%
2^{2TR_{m}}$, where the factor of $2T$ is because each round of information
exchange consists of two length-$T$ time-slots. At user $A$, the estimated
message of user $B$, denoted by $\hat{W}_{B}$, is obtained from the received
signal $\mathbf{Y}_{A}$ and the perfect knowledge of the self message $W_{A}$%
. The decoding operation at user $B$ is similar. The error probability is
defined as $P_{e}\triangleq \Pr \{\hat{W}_{A}\neq W_{A}$\ or $\hat{W}%
_{B}\neq W_{B}\}$\textbf{. }We say that a rate-pair $(R_{A},R_{B})$ is
achievable if the error probability $P_{e}$\ vanishes as $T$ tends to
infinity. The achievable rate-region is defined as the closure of all
possible achievable rate-pairs.

\subsection{Capacity Upper Bound}

Here we briefly describe a capacity upper bound of the MIMO TWRC. Let $%
\mathbf{Q}_{m}\triangleq \frac{1}{T}E\left[ \mathbf{X}_{m}\mathbf{X}%
_{m}^{\dagger }\right] ,m\in \left\{ A,B,R\right\},$ be the input covariance
matrices. For given \{$\mathbf{Q}_{A}$, $\mathbf{Q}_{B}$, $\mathbf{Q}_{R}$\}
satisfying tr$(\mathbf{Q}_{m})\leq P_{m},m\in \{A,B,R\} $, the achievable
rate-pairs $(R_{A},R_{B})$ of the MIMO TWRC is upper bounded as \cite%
{HJYangIT11}
\begin{subequations}
\label{rate}\begin{eqnarray}
R_{A} &\leq &\min \left\{ R_{A}^{UL}\left( \mathbf{Q}_{A}\right)
,R_{A}^{DL}\left( \mathbf{Q}_{R}\right) \right\}   \\
R_{B} &\leq &\min \left\{ R_{B}^{UL}\left( \mathbf{Q}_{B}\right)
,R_{B}^{DL}\left( \mathbf{Q}_{R}\right) \right\}
\end{eqnarray}%
\end{subequations}
where%
\begin{subequations}\label{Outer Bound}
\begin{eqnarray}
R_{m}^{UL}\left( \mathbf{Q}_{m}\right) &=&\frac{1}{2}\log \left\vert \mathbf{%
I}_{n_{R}}+\frac{1}{N_{0}}\mathbf{H}_{mR}\mathbf{Q}_{m}\mathbf{H}%
_{mR}^{\dagger }\right\vert ,m\in \left\{ A,B\right\}   \\
R_{A}^{DL}\left( \mathbf{Q}_{R}\right) &=&\frac{1}{2}\log \left\vert \mathbf{%
I}_{n_{B}}+\frac{1}{N_{0}}\mathbf{H}_{RB}\mathbf{Q}_{R}\mathbf{H}%
_{RB}^{\dagger }\right\vert , \\
R_{B}^{DL}\left( \mathbf{Q}_{R}\right) &=&\frac{1}{2}\log \left\vert \mathbf{%
I}_{n_{A}}+\frac{1}{N_{0}}\mathbf{H}_{RA}\mathbf{Q}_{R}\mathbf{H}%
_{RA}^{\dagger }\right\vert .
\end{eqnarray}%
\end{subequations}
Here, the superscripts \textquotedblleft UL\textquotedblright\ and
\textquotedblleft DL\textquotedblright\ respectively represent uplink and
downlink, and the factor of 1/2 is due to the two time-slots used for each
round of information exchange.

A capacity-region outer bound is defined as the closure of the upper-bound
rate-pairs in (\ref{rate}). This outer bound can be determined by optimizing
$\mathbf{Q}_{A}$, $\mathbf{Q}_{B}$, and $\mathbf{Q}_{R}$, as detailed in
\cite{HJYangIT11} and \cite{YangIT11}. The goal of this paper is to develop
a communication strategy that can approach this outer bound.

\section{Relaying Strategies for TWRCs with Single-Antenna Users}

In this section, we study efficient communications over TWRCs with
single-antenna users and a multi-antenna relay, i.e., $n_{A}=n_{B}=1$ and $%
n_{R}\geq 1$. The results developed in this section will be used in our
studies on general MIMO TWRCs later.

\subsection{Relaying Strategies: Complete Decoding versus PNC}

For the case of single-antenna users, the channel model of the uplink phase
in (\ref{YR}) can be simplified as%
\begin{equation}
\mathbf{Y}_{R}=\mathbf{h}_{AR}\mathbf{x}_{A}^{T}+\mathbf{h}_{BR}\mathbf{x}%
_{B}^{T}+\mathbf{Z}_{R}  \label{Vector}
\end{equation}%
where $\mathbf{h}_{mR}\in \mathcal{%
\mathbb{C}
}^{n_{R}\times 1}$ is the reduced version of $\mathbf{H}_{mR}$, and $\mathbf{%
x}_{m}\in \mathcal{%
\mathbb{C}
}^{T\times 1}$ is the transmit signal vector of user $m$, with the $i$th
entry of $\mathbf{x}_{m}$ being the signal transmitted at the $i$th time
interval, $m\in \left\{ A,B\right\} $. Upon receiving $\mathbf{Y}_{R}$, the
relay generates a network-coded message following the spirit of network
coding. This network-coded message will be forwarded to the two users in the
downlink phase.

The relay operation is detailed as follows.
In the uplink phase, the signal direction of user $m$ is given by $\mathbf{h}%
_{mR},$ $m\in \left\{ A,B\right\} $. On one hand, if $\mathbf{h}_{AR}$ and $%
\mathbf{h}_{BR}$ are orthogonal, both messages of the two users can be
decoded free of interference from each other. The recovered messages of the
two users are then network-coded and forwarded to the two users. We refer to
this first strategy as the \textit{complete-decoding (CD) strategy}. On the
other hand, if $\mathbf{h}_{AR}$ and $\mathbf{h}_{BR}$ turn out to be
parallel (i.e., in a same direction), then it is advantageous to compute a
linear function of $\mathbf{x}_{A}$ and $\mathbf{x}_{B}$, referred to as a
network-coded message, instead of completely decoding both $\mathbf{x}_{A}$
and $\mathbf{x}_{B}$. We refer to this second strategy as the \textit{PNC
strategy}.

In general, the following strategy can be adopted: if $\mathbf{h}_{AR}$ and $%
\mathbf{h}_{BR}$ tend to be orthogonal, the complete-decoding strategy is
applied; if $\mathbf{h}_{AR}$ and $\mathbf{h}_{BR}$ tend to be parallel, the
PNC strategy is applied. The selection between these two strategies is based
on their achievable rates, as described below.

\subsubsection{The Complete-Decoding Strategy}

For complete-decoding, the uplink channel in (\ref{Vector}) becomes a
multiple-access (MAC) channel. Let $R_{m}^{CD},m\in \left\{ A,B\right\} $,
be the rate of user $m$ for the complete-decoding strategy. Then, the uplink
rate-region of the complete-decoding strategy, denoted by $\mathcal{R}%
_{{}}^{CD}$, is given by%
\begin{subequations}\label{MAC}
\begin{eqnarray}
R_{A}^{CD}+R_{B}^{CD} &\leq &\frac{1}{2}\log \left\vert \mathbf{I}%
_{n_{R}}+\dsum\limits_{m\in \left\{ A,B\right\} }\frac{P_{m}}{N_{0}}\mathbf{h%
}_{mR}\mathbf{h}_{mR}^{\dagger }\right\vert   \\
R_{m}^{CD} &\leq &\frac{1}{2}\log \left( 1+\frac{P_{m}}{N_{0}}\mathbf{h}%
_{mR}^{\dagger }\mathbf{h}_{mR}\right) ,m\in \left\{ A,B\right\} .
\end{eqnarray}%
\end{subequations}
which follows from the well-known MAC capacity region \cite{Cover91}.

\subsubsection{The PNC Strategy}

For the PNC strategy in \cite{NamIT10}\cite{WilsonIT10}, it is required that
the two user-signals lie in a same spatial direction. This is not guaranteed
here due to the availability of multiple antennas at the relay. To
facilitate PNC, we next propose a projection-based method as follows. The
signals from the two users are first projected to a common direction,
denoted by a unit vector $\mathbf{p}\in \mathcal{%
\mathbb{C}
}^{n_{R}\times 1}$. The choice of $\mathbf{p}$ will be discussed in the next
subsection. This projection operation creates an aligned scalar channel
given by%
\begin{equation}
\mathbf{p}^{\dag }\mathbf{Y}_{R}=\mathbf{p}^{\dag }\mathbf{h}_{AR}\mathbf{x}%
_{A}^{T}+\mathbf{p}^{\dag }\mathbf{h}_{BR}\mathbf{x}_{B}^{T}+\mathbf{p}%
^{\dag }\mathbf{Z}_{R}
\end{equation}%
with the effective channel coefficients given by $\mathbf{p}^{\dag }\mathbf{h%
}_{AR}$ and $\mathbf{p}^{\dag }\mathbf{h}_{BR}$.

From \cite{HJYangIT11}, if the sum of the two users' codewords, i.e., $%
\mathbf{p}^{\dag }\mathbf{h}_{AR}\mathbf{x}_{A}^{T}+\mathbf{p}^{\dag }%
\mathbf{h}_{BR}\mathbf{x}_{B}^{T}$, is required to be decoded, an achievable
rate-pair ($R_{A}^{PNC},R_{B}^{PNC}$) of the uplink phase is given by%
\begin{equation}
R_{m}^{PNC}=\frac{1}{2}\left[ \log \left( \frac{Q_{m}\left\vert \mathbf{p}%
^{\dag }\mathbf{h}_{mR}\right\vert ^{2}}{N_{0}}\right) \right] ^{+},m\in
\left\{ A,B\right\} ,  \label{R_PNC}
\end{equation}%
where $Q_{m}=\frac{1}{T}E[\mathbf{x}_{m}^{\dagger }\mathbf{x}_{m}]$
represents the transmission power of user $m$. If $\mathbf{p}^{\dag }\mathbf{%
h}_{AR}\mathbf{x}_{A}^{T}+\mathbf{p}^{\dag }\mathbf{h}_{BR}\mathbf{x}%
_{B}^{T} $ is not necessarily decoded, the above rate-pair can be further
improved by using the lattice-modulo operation and minimum mean-square error
(MMSE) scaling \cite{NamIT10}\cite{ErezIT04}, with the achievable rate-pair
given by%
\begin{equation}
R_{m}^{PNC}=\frac{1}{2}\left[ \log \left( \frac{Q_{m}\left\vert \mathbf{p}%
^{\dag }\mathbf{h}_{mR}\right\vert ^{2}}{Q_{A}\left\vert \mathbf{p}^{\dag }%
\mathbf{h}_{AR}\right\vert ^{2}+Q_{B}\left\vert \mathbf{p}^{\dag }\mathbf{h}%
_{BR}\right\vert ^{2}}+\frac{Q_{m}\left\vert \mathbf{p}^{\dag }\mathbf{h}%
_{mR}\right\vert ^{2}}{N_{0}}\right) \right] ^{+},m\in \left\{ A,B\right\} .
\label{R_PNC2}
\end{equation}%
Notice that (\ref{R_PNC}) and (\ref{R_PNC2}) become identical and both
approaches the uplink channel capacity at high SNR.

The uplink rate-region of the PNC scheme is given by%
\begin{equation}
\mathcal{R}_{{}}^{PNC}\triangleq \left\{ (R_{A},R_{B})|R_{m}\leq
R_{m}^{PNC},Q_{m}\leq P_{m},m\in \left\{ A,B\right\} ,\mathbf{p}^{\dag }%
\mathbf{p=1}\right\} .  \label{Eq PNC SIMO}
\end{equation}%
The boundary of $\mathcal{R}_{{}}^{PNC}$ can be found by optimizing $Q_{A}$,
$Q_{B}$, and $\mathbf{p}$, as detailed in the next subsection.

\subsection{Optimization of Projection Direction}

Here we focus on the rate-pair given in (\ref{R_PNC}).\footnote{%
The treatment for the rate-pair in (\ref{R_PNC2}) is similar, and thus
omitted.} As achievable rate-regions are convex, the boundary points of $%
\mathcal{R}_{{}}^{PNC}$ can be determined by solving the weighted sum-rate
maximization problem:%
\begin{subequations}\label{P1}
\begin{eqnarray}
\text{maximize} &&w_{A}R_{A}^{PNC}+w_{B}R_{B}^{PNC} \\
\text{subject to} &&||\mathbf{p}||=1,Q_{m}\leq P_{m},m\in \left\{ A,B\right\}
\end{eqnarray}%
\end{subequations}
where $w_{A}$ and $w_{B}$ are arbitrary nonnegative weighting coefficients.
By inspecting (\ref{R_PNC}), the maximum of (\ref{P1}) is achieved at $%
Q_{m}=P_{m},m\in \left\{ A,B\right\} $. Thus, we only need to optimize $%
\mathbf{p}$.

Suppose that $R_{A}^{PNC}=0$ (or $R_{B}^{PNC}=0$). Then, from (\ref{R_PNC}),
the optimal $\mathbf{p}$ is trivially taken as $\frac{\mathbf{h}_{BR}}{%
\left\Vert \mathbf{h}_{BR}\right\Vert }$ (or $\frac{\mathbf{h}_{AR}}{%
\left\Vert \mathbf{h}_{AR}\right\Vert }$). Thus, we focus on the case of $%
R_{m}^{PNC}>0,m\in \left\{ A,B\right\} $. In this case, this weighted
sum-rate maximization problem is equivalent to maximizing%
\begin{equation}
\max_{||\mathbf{p}||=1}w_{A}\log \left( \left\vert \mathbf{p}^{\dag }\mathbf{%
h}_{AR}\right\vert ^{2}\right) +w_{B}\log \left( \left\vert \mathbf{p}^{\dag
}\mathbf{h}_{BR}\right\vert ^{2}\right)
\end{equation}%
or equivalently%
\begin{equation*}
\max_{||\widetilde{\mathbf{p}}||=1}\left| \widetilde{\mathbf{h}}_{AR}^{T}\widetilde{\mathbf{%
p}}\right| ^{2w_{A}}\left| \widetilde{\mathbf{h}}_{BR}^{T}\widetilde{\mathbf{p}}\right| ^{2w_{B}}
\end{equation*}%
where $\widetilde{\mathbf{p}}=[\func{Re}[\mathbf{p}^{T}],\func{Im}[\mathbf{p}%
^{T}]]^{T}$ and $\widetilde{\mathbf{h}}_{mR}=[\func{Re}[\mathbf{h}_{mR}^{T}],%
\func{Im}[\mathbf{h}_{mR}^{T}]]^{T}$, $m\in \left\{ A,B\right\} $.
By setting the derivative of the Lagrangian with respect to $\textbf{p}$ to zero, the optimal $\textbf{p}$ satisfies%
\begin{equation}
\alpha \widetilde{\mathbf{p}}=\frac{w_{A}\widetilde{\mathbf{h}}_{AR}}{%
\widetilde{\mathbf{p}}^{T}\widetilde{\mathbf{h}}_{AR}}+\frac{w_{B}\widetilde{%
\mathbf{h}}_{BR}}{\widetilde{\mathbf{p}}^{T}\widetilde{\mathbf{h}}_{BR}}, \notag
\end{equation}%
where $\alpha$ is a scaling factor. Then, with some straightforward algebra, we obtain the optimal projection direction given by
\begin{equation}
\widetilde{\mathbf{p}}^{opt}=\gamma \left( \frac{\widetilde{\mathbf{h}}_{AR}%
}{\left\Vert \widetilde{\mathbf{h}}_{AR}\right\Vert }+\beta \frac{\widetilde{%
\mathbf{h}}_{BR}}{\left\Vert \widetilde{\mathbf{h}}_{BR}\right\Vert }\right)
,  \label{d}
\end{equation}%
where%
\begin{equation}
\beta =\frac{\text{sign}(\widetilde{\mathbf{h}}_{AR}^{T}\widetilde{\mathbf{h}%
}_{BR})}{2}\left( \sqrt{\left( \frac{\widetilde{\mathbf{h}}_{AR}^{T}%
\widetilde{\mathbf{h}}_{BR}\left( 1-\frac{w_{B}}{w_{A}}\right) }{\left\Vert
\widetilde{\mathbf{h}}_{AR}\right\Vert \left\Vert \widetilde{\mathbf{h}}%
_{BR}\right\Vert }\right) ^{2}+4\frac{w_{B}}{w_{A}}}-\frac{\left\vert
\widetilde{\mathbf{h}}_{AR}^{T}\widetilde{\mathbf{h}}_{BR}\right\vert \left(
1-\frac{w_{B}}{w_{A}}\right) }{\left\Vert \widetilde{\mathbf{h}}%
_{AR}\right\Vert \left\Vert \widetilde{\mathbf{h}}_{BR}\right\Vert }\right)
\end{equation}%
and $\gamma $ is a scaling factor to ensure $||\widetilde{\mathbf{p}}^{opt}||$ $=1$. Particularly, for
the sum-rate case, i.e., $w_{A}=w_{B}=1$, the optimal projection direction $%
\widetilde{\mathbf{p}}$ is just the angular bisector of $\widetilde{\mathbf{h%
}}_{AR}$ and $\widetilde{\mathbf{h}}_{BR}$ if $\widetilde{\mathbf{h}}%
_{AR}^{T}\widetilde{\mathbf{h}}_{BR}>0$, or the angular bisector of $%
\widetilde{\mathbf{h}}_{AR}$ and $-\widetilde{\mathbf{h}}_{BR}$ if $%
\widetilde{\mathbf{h}}_{AR}^{T}\widetilde{\mathbf{h}}_{BR}<0$. By varying $%
w_{A}$ and $w_{B}$, $\mathcal{R}_{{}}^{PNC}$ can be determined.

\subsection{The Overall Scheme}

For the uplink phase, the achievable rate-regions $\mathcal{R}_{{}}^{CD}$
and $\mathcal{R}_{{}}^{PNC}$, for certain channel realizations of $\mathbf{h}%
_{AR}$ and $\mathbf{h}_{BR}$, are depicted in Fig. \ref{SIMO}. The overall
uplink rate-region, denoted by $\mathcal{R}_{{}}^{UL}$, is given by the
convex hull of $\mathcal{R}_{{}}^{CD}$ and $\mathcal{R}_{{}}^{PNC}$. In the
overall scheme, the relay will select between the complete-decoding and PNC
strategies for a larger achievable rate-region, according to (\ref{MAC}) and
(\ref{Eq PNC SIMO}).

For the downlink phase, the achievable rate-region is determined as follows.
For the complete-decoding strategy, the relay jointly re-encode the decoded
messages $\mathbf{x}_{A}$ and $\mathbf{x}_{B}$, and forward the resulting
codeword to the two users in the downlink. For the PNC strategy, the relay
forward the lattice-modulo of $\mathbf{p}^{\dag }\mathbf{h}_{AR}\mathbf{x}%
_{A}^{T}+\mathbf{p}^{\dag }\mathbf{h}_{BR}\mathbf{x}_{B}^{T}$, referred to
as a network-coded message, to the two users. Then each user recovers the
message of the other user with the help of the perfect knowledge of self
message. From \cite{HJYangIT11}-\cite{KhinaISIT11}, the downlink rate-region
$\mathcal{R}_{{}}^{DL}$ for the two strategies are the same and given by
\begin{equation}
\mathcal{R}_{{}}^{DL}\triangleq \left\{ (R_{A},R_{B})|R_{A}\leq
R_{A}^{DL},R_{B}\leq R_{B}^{DL}\right\}
\end{equation}%
with%
\begin{equation}
R_{A}^{DL}=\frac{1}{2}\log \left( 1+\frac{P_{R}}{N_{0}}\mathbf{h}%
_{RB}^{\dagger }\mathbf{h}_{RB}\right) \text{ and }R_{B}^{DL}=\frac{1}{2}%
\log \left( 1+\frac{P_{R}}{N_{0}}\mathbf{h}_{RA}^{\dagger }\mathbf{h}%
_{RA}\right) .
\end{equation}%
Finally, an achievable rate-region of the overall scheme is the intersection
of the uplink rate-region $\mathcal{R}_{{}}^{UL}$ and the downlink
rate-region $\mathcal{R}_{{}}^{DL}$.

\section{Space-Division Approach for MIMO TWRCs}

In the preceding section, we have studied the design of relaying strategies
for TWRCs with single-antenna users. We have shown how to exploit the
benefits of the complete-decoding and PNC strategies. In this section, we
proceed to study the general case of {$n_{A}\geq 1$, $n_{B}\geq 1$}. We
propose a new space-division based network-coding scheme, as a
generalization for the case of single-antenna users.

\subsection{Motivations}

What motivates the proposed space-division approach is the following
property of $\mathbf{H}_{AR}$ and $\mathbf{H}_{BR}$. Denote by $\mathcal{C}(\mathbf{H}_{AR})$ and $\mathcal{C}(\mathbf{H}%
_{BR})$ the columnspaces of the uplink channel matrices $\mathbf{H}_{AR}$ and $\mathbf{H}_{BR}$,
respectively. In general, we can partition the columnspace $\mathcal{C}(\mathbf{H}%
_{AR})$ $\in \mathcal{%
\mathbb{C}
}^{n_{R}}$ as the direct sum\footnote{%
Let $\mathcal{S}$ be a vector space, and let $\mathcal{S}_{1},\mathcal{S}%
_{2},...,\mathcal{S}_{n}$ be subspaces of $\mathcal{S}$. $\mathcal{S}$
is defined to be a direct sum of $\mathcal{S}_{1},\mathcal{S}_{2},...,%
\mathcal{S}_{n}$ when $\mathcal{S}_{1},\mathcal{S}_{2},...,\mathcal{S}_{n}$
are mutually orthogonal and for every vector $\mathbf{x}$ in $\mathcal{S}$,
there is $\mathbf{x}_{i}$ in $\mathcal{S}_{i},i=1,2,...,n,$ such that $%
\mathbf{x=}\dsum\limits_{i=1}^{n}\mathbf{x}_{i}$.} of three orthogonal
subspaces: a subspace $\mathcal{S}_{A\parallel B}$ that is parallel to $%
\mathcal{C}(\mathbf{H}_{BR})$, i.e., any vector in $\mathcal{S}_{A\parallel
B}$ belongs to $\mathcal{C}(\mathbf{H}_{BR})$; a subspace $\mathcal{S}%
_{A\nparallel B}$ that is neither parallel nor orthogonal to $\mathcal{C}(%
\mathbf{H}_{BR})$; and a subspace $\mathcal{S}_{A\perp B}$ that is
orthogonal to $\mathcal{C}(\mathbf{H}_{BR})$. Similarly, $\mathcal{C}(%
\mathbf{H}_{BR})$ is the direct sum of three orthogonal subspaces $\mathcal{S%
}_{B||A},\mathcal{S}_{B\nparallel A},$ and $\mathcal{S}_{B\perp A}$. Note
that $\mathcal{S}_{A\parallel B}=\mathcal{S}_{B||A}$ since both represent
the \textit{common space} of $\mathcal{C}(\mathbf{H}_{AR})$ and $\mathcal{C}(%
\mathbf{H}_{BR})$.

In $\mathcal{S}_{A\parallel B}$, the signal directions of the two users can
be efficiently aligned to a same set of directions, providing a platform to
carry out PNC, as in \cite{HJYangIT11}-\cite{KhinaISIT11}. On the other
hand,\ in $\mathcal{S}_{A\perp B}$ or $\mathcal{S}_{B\perp A}$, the two
users do not interfere with each other, hence the complete-decoding strategy
can be employed. The above treatments are similar to those for the case of
single-antenna users, as discussed in the preceding section. What remains is
the treatment for the signals in $\mathcal{S}_{A\nparallel B}$ and $\mathcal{%
S}_{B\nparallel A}$ that are neither parallel nor orthogonal. Heuristically,
some channel directions in $\mathcal{S}_{A\nparallel B}$ and $\mathcal{S}%
_{B\nparallel A}$ may be nearly parallel to each other. For these channel
directions, the PNC strategy is preferable for the related spatial streams.
On the other hand, some channel directions in $\mathcal{S}_{A\nparallel B}$
and $\mathcal{S}_{B\nparallel A}$ may be nearly orthogonal to each other.
Then, the complete-decoding strategy is preferable. The main challenge lies
in how to identify those nearly parallel/orthogonal channel directions. To
this end, we next propose a new joint channel decomposition.

\subsection{Joint Channel Decomposition}

Let the compact singular value decomposition (SVD) of $\mathbf{H}_{mR}$ be%
\begin{equation}
\mathbf{H}_{mR}=\mathbf{U}_{m}\mathbf{\Delta }_{m}\mathbf{V}_{m}^{\dag },%
\mathbf{\ }m\in \left\{ A,B\right\}  \label{Channel_Decomp}
\end{equation}%
where $\mathbf{U}_{m}$ is an $n_{R}$-by-$n_{m}$ orthonormal matrix with $%
\mathbf{U}_{m}^{\dag }\mathbf{U}_{m}=\mathbf{I}_{n_{m}}$.


Denote by $\lambda _{i}$ the $i$th eigenvalue of the matrix $\mathbf{U}_{A}%
\mathbf{U}_{A}^{\dag }+\mathbf{U}_{B}\mathbf{U}_{B}^{\dag }$, and by $%
\mathbf{u}_{i}$ the corresponding eigenvector. Without loss of generality,
we arrange \{$\lambda _{i}$\} in the descending order. As the eigenvalues of
$\mathbf{U}_{m}\mathbf{U}_{m}^{\dag }$ are either 1 or 0, the eigenvalues \{$%
\lambda _{i}$\} are valued between $0$ and $2$. Note that $\lambda
_{i}=2$ implies that $\mathbf{U}_{A}\mathbf{U}_{A}^{\dag }\mathbf{u}_{i}%
\mathbf{=u}_{i}$ and $\mathbf{U}_{B}\mathbf{U}_{B}^{\dag }\mathbf{u}_{i}%
\mathbf{=u}_{i}$. This means that $\mathbf{u}_{i}$ is in the common space $%
\mathcal{S}_{A\parallel B}$ of $\mathbf{H}_{AR}$ and $\mathbf{H}_{BR}$. Also note
that $\lambda _{i}=1$ implies \{$\mathbf{U}_{A}\mathbf{U}_{A}^{\dag }\mathbf{u}_{i}\mathbf{=u}_{i}$,
 $\mathbf{U}_{B}\mathbf{U}_{B}^{\dag }\mathbf{u}_{i}\mathbf{=0}$\} or \{$%
\mathbf{U}_{A}\mathbf{U}_{A}^{\dag }\mathbf{u}_{i}\mathbf{=0}$, $\mathbf{U%
}_{B}\mathbf{U}_{B}^{\dag }\mathbf{u}_{i}\mathbf{=u}_{i}$\} (cf., Theorem 4.3.4 in \cite{HornTextbook}). That is, $\mathbf{%
u}_{i}$ is in $\mathcal{S}_{A\perp B}$ (or $\mathcal{S}_{B\perp A}$) which
is orthogonal to the space spanned by $\mathbf{H}_{BR}$ (or $\mathbf{H}_{AR}$%
). In addition, the number of eigenvalues between 1 and 2 is
the same as that between 0 and 1, as we will see later.

Let $k$ be the number of eigenvalues of $\mathbf{U}_{A}\mathbf{U}_{A}^{\dag
}+\mathbf{U}_{B}\mathbf{U}_{B}^{\dag }$ equal to 2; $l$ be the number of
eigenvalues between 1 and 2; $d_{A}$ be the number of eigenvalues equal to 1
with \{$\mathbf{U}_{A}\mathbf{U}_{A}^{\dag }\mathbf{u}_{i}\mathbf{=u}_{i}$, $%
\mathbf{U}_{B}\mathbf{U}_{B}^{\dag }\mathbf{u}_{i}\mathbf{=0}$\}; $d_{B}$ be
the number of eigenvalues equal to 1 with \{$\mathbf{U}_{A}\mathbf{U}_{A}^{\dag }%
\mathbf{u}_{i}\mathbf{=0}$, $\mathbf{U}_{B}\mathbf{U}_{B}^{\dag }\mathbf{u}%
_{i}\mathbf{=u}_{i}$\}. Also let $\mathbf{U}\in $ $\mathcal{%
\mathbb{C}
}^{n_{R}\times (n_{A}+n_{B}-k)}$ be a matrix with the columns consisting of
the eigenvectors corresponding to the $n_{A}+n_{B}-k$ largest eigenvalues of
$\mathbf{U}_{A}\mathbf{U}_{A}^{\dag }+\mathbf{U}_{B}\mathbf{U}_{B}^{\dag }$
(as specified in (\ref{Def_U})). Due to the orthogonality of the
eigenvectors, $\mathbf{U}$ is orthonormal, i.e., $\mathbf{U}_{{}}^{\dag }%
\mathbf{U}=\mathbf{I}_{n_{A}+n_{B}-k}$.

We are now ready to present the joint channel decomposition, with the proof
given in Appendix A.

\begin{theo}
\label{Theorem 1}The channel matrices\textit{\ }$\mathbf{H}_{AR}$ and $%
\mathbf{H}_{BR}$\ can be jointly decomposed as%
\begin{equation}
\mathbf{H}_{mR}=\mathbf{UD}_{m}\mathbf{G}_{m},\mathbf{\ }m\in \left\{
A,B\right\}  \label{Hm}
\end{equation}%
where $\mathbf{G}_{m}\in $ $\mathcal{%
\mathbb{C}
}^{n_{m}\times n_{m}}$ is a square matrix, and $\mathbf{D}_{m}\in \mathcal{%
\mathbb{C}
}^{(n_{A}+n_{B}-k)\times n_{m}},m\in \left\{ A,B\right\} ,$ are orthonormal
matrices with a block-diagonal structure given by%
\begin{subequations}
\begin{equation}
\mathbf{D}_{A}=\left[
\begin{array}{ccc}
\mathbf{I}_{k} & \mathbf{0} & \mathbf{0} \\
\mathbf{0} & \mathbf{E}_{A} & \mathbf{0} \\
\mathbf{0} & \mathbf{0} & \mathbf{I}_{d_{A}} \\
\mathbf{0} & \mathbf{0} & \mathbf{0}%
\end{array}%
\right] \text{ and }\mathbf{D}_{B}=\left[
\begin{array}{ccc}
\mathbf{I}_{k} & \mathbf{0} & \mathbf{0} \\
\mathbf{0} & \mathbf{E}_{B} & \mathbf{0} \\
\mathbf{0} & \mathbf{0} & \mathbf{0} \\
\mathbf{0} & \mathbf{0} & \mathbf{I}_{d_{B}}%
\end{array}%
\right]  \label{Dm}
\end{equation}%
with
\begin{equation}
\mathbf{E}_{m}=\left[
\begin{array}{cccc}
\mathbf{e}_{m;k+1} & \mathbf{0} & \cdots & \mathbf{0} \\
\mathbf{0} & \mathbf{e}_{m;k+2} & \ddots & \vdots \\
\vdots & \ddots & \ddots & \mathbf{0} \\
\mathbf{0} & \cdots & \mathbf{0} & \mathbf{e}_{m;k+l}%
\end{array}%
\right] \in \mathcal{%
\mathbb{C}
}^{2l\times l}\text{, }  \label{E}
\end{equation}%
\begin{equation}
\mathbf{e}_{A;i}=\left[
\begin{array}{c}
\sqrt{\frac{\lambda _{i}}{2}} \\
\sqrt{\frac{2-\lambda _{i}}{2}}%
\end{array}%
\right] \text{ and }\mathbf{e}_{B;i}=\left[
\begin{array}{c}
\sqrt{\frac{\lambda _{i}}{2}} \\
-\sqrt{\frac{2-\lambda _{i}}{2}}%
\end{array}%
\right] .  \label{Direction}
\end{equation}
\end{subequations}
\end{theo}

\begin{rema}
From (\ref{Hm}), we see that $\mathcal{C}(\mathbf{U)}$ is the overall
columnspace of the two channel matrices, i.e., $\mathcal{C}(\mathbf{U)}=%
\mathcal{C}([\mathbf{H}_{AR}$ $\mathbf{H}_{BR}]\mathbf{)}$. Moreover, $%
\mathbf{UD}_{m}$ specifies the columnspace of $\mathbf{H}_{m}$, i.e., $%
\mathcal{C}(\mathbf{UD}_{m})=\mathcal{C}(\mathbf{H}_{m})$. Note that $%
\mathbf{UD}_{m}$ is orthonormal, as $\mathbf{D}_{m}^{\dag }\mathbf{U}^{\dag }%
\mathbf{UD}_{m}=\mathbf{I}_{n_{m}}$, $m\in \left\{ A,B\right\} $. Therefore,
the columns of $\mathbf{UD}_{m}$ give an orthogonal basis of $\mathcal{C}(%
\mathbf{H}_{mR})$, with the coordinates of $\mathbf{H}_{mR}$ specified in $%
\mathbf{G}_{m}$.
\end{rema}

\begin{rema}
The column structures of $\mathbf{UD}_{A}$ and $\mathbf{UD}_{B}$ are
explained as follows. In the first place, we note that $\mathbf{UD}_{A}$ and
$\mathbf{UD}_{B}$ share the same first $k$ columns. Thus, the first $k$
columns of $\mathbf{UD}_{A}$ span $\mathcal{S}_{A\parallel B}$, i.e., the
common space of $\mathcal{C}(\mathbf{H}_{AR})$ and $\mathcal{C}(\mathbf{H}%
_{BR})$. Second, from (\ref{Dm}), the last $d_{A}$ columns of $\mathbf{UD}%
_{A}$ (obtained from multiplying $\mathbf{U}$ with the third block column of
$\mathbf{D}_{A}$) are orthogonal to $\mathbf{UD}_{B}$. Hence, these columns
of $\mathbf{UD}_{A}$ span the subspace $\mathcal{S}_{A\perp B}$, i.e., the
subspace orthogonal to $\mathcal{C}(\mathbf{H}_{BR})$. Third, the remaining $%
l$ columns of $\mathbf{UD}_{A}$ span the subspace $\mathcal{S}_{A\nparallel
B}$, by noting the facts that $\mathcal{C}(\mathbf{H}_{AR}\mathbf{)}=%
\mathcal{C}(\mathbf{UD}_{A}\mathbf{)}$ and that $\mathcal{C}(\mathbf{H}%
_{AR}) $ is the direct-sum of three orthogonal subspaces $\mathcal{S}%
_{A\parallel B} $, $\mathcal{S}_{A\nparallel B}$, and $\mathcal{S}_{A\perp
B} $. Similarly, the first $k$ columns of $\mathbf{UD}_{B}$ span $\mathcal{S}%
_{A\parallel B}$, the next $l$ columns span $\mathcal{S}_{B\nparallel A}$,
and the last $d_{B}$ columns span $\mathcal{S}_{B\perp A}$. Recall that $%
\mathcal{C}(\mathbf{H}_{AR})$ is the direct sum of $\mathcal{S}_{A\parallel
B}$, $\mathcal{S}_{A\nparallel B}$, and $\mathcal{S}_{A\perp B}$, and that $%
\mathcal{C}(\mathbf{H}_{BR})$ is the direct sum of $\mathcal{S}_{B\parallel
A}$, $\mathcal{S}_{B\nparallel A}$, and $\mathcal{S}_{B\perp A}$. Thus, the
dimensions of these subspaces have the following relationship:
\begin{equation}
k+l+p_{m}=n_{m},m\in \left\{ A,B\right\} \text{.}
\end{equation}%
We summarize the geometrical meanings of the aforementioned subspaces and
their dimensions as follows.%
\begin{equation*}
\begin{tabular}{lll}
Subspace & Dimension & Property \\
$\mathcal{S}_{A\parallel B}$ & $k$ & $\text{common space of }\mathcal{C}(%
\mathbf{H}_{AR})\text{ and }\mathcal{C}(\mathbf{H}_{BR})$ \\
$\mathcal{S}_{A\nparallel B}$ & $l$ & not parallel/orthogonal to$\text{ }%
\mathcal{C}(\mathbf{H}_{BR})$ \\
$\mathcal{S}_{B\nparallel A}$ & $l$ & not parallel/orthogonal to$\text{ }%
\mathcal{C}(\mathbf{H}_{AR})$ \\
$\mathcal{S}_{A\perp B}$ & $d_{A}$ & orthogonal to$\text{ }\mathcal{C}(%
\mathbf{H}_{BR})$ \\
$\mathcal{S}_{B\perp A}$ & $d_{B}$ & orthogonal to$\text{ }\mathcal{C}(%
\mathbf{H}_{AR})$%
\end{tabular}%
\end{equation*}
\end{rema}

Let $\mathbf{v}_{m;i}$ be the $i$th column of $\mathbf{UD}_{m}$, $m\in
\left\{ A,B\right\} $. We refer to $\mathbf{v}_{A;i}$ and $\mathbf{v}_{B;i}$
as the $i$th \textit{channel direction pair}. Here, $\mathbf{v}%
_{A;i}^{^{\dag }}\mathbf{v}_{B;i}=1$ means that $\mathbf{v}_{A;i}$ and $%
\mathbf{v}_{B;i}$ are parallel, and $\mathbf{v}_{A;i}^{^{\dag }}\mathbf{v}%
_{B;i}=0$ means that they are orthogonal. Thus, $\mathbf{v}_{A;i}^{^{\dag }}%
\mathbf{v}_{B;i}$ can be regarded as a measure of the \textit{degree of
orthogonality }of $\mathbf{v}_{A;i}$ and $\mathbf{v}_{B;i}$. In the
following corollary, the degree of orthogonality of each channel direction
pair ($\mathbf{v}_{A;i}$, $\mathbf{v}_{B;i}$) is determined by the magnitude
of $\lambda _{i}$, i.e., the $i$th eigenvalue of $\mathbf{U}_{A}\mathbf{U}%
_{A}^{\dag }+\mathbf{U}_{B}\mathbf{U}_{B}^{\dag }$.

\begin{coro}
\label{Corollary 1}For $i=1,...,k+l$, the degree of orthogonality of the $i$%
th channel direction pair ($\mathbf{v}_{A,i}$, $\mathbf{v}_{B,i}$) is given
by $\mathbf{v}_{A,i}^{^{\dag }}\mathbf{v}_{B,i}=\lambda _{i}-1$.
\end{coro}

\begin{proof}
For $i=1,...,k$, we see from (\ref{Dm}) that $\lambda _{i}=2$ and $\mathbf{v}%
_{A;i}=\mathbf{v}_{B;i}$, and so $\mathbf{v}_{A;i}^{\dag }\mathbf{v}%
_{B;i}=\lambda _{i}-1$. For $i=k+1,...,k+l$, from (\ref{Dm}) and (\ref{E}),
the $i$th column of $\mathbf{UD}_{m}$ is given by
\begin{equation}
\mathbf{v}_{m;i}=\left[
\begin{array}{cc}
\widetilde{\mathbf{u}}_{2i-k-1} & \widetilde{\mathbf{u}}_{2i-k}%
\end{array}%
\right] \mathbf{e}_{m;i},\mathbf{\ }m\in \left\{ A,B\right\} ,\text{ }
\end{equation}%
where $\widetilde{\mathbf{u}}_{i}$ represents the $i$th column of $\mathbf{U}
$. Then, we obtain $\mathbf{v}_{A;k+i}^{\dag }\mathbf{v}_{B;k+i}=\mathbf{e}%
_{A;k+i}^{\dag }\mathbf{e}_{B;k+i}=\lambda _{i}-1$, where the first step
utilizes the fact that $\mathbf{U}$ is orthonormal, and the second step
follows from (\ref{Direction}).
\end{proof}

\begin{coro}
\label{Corollary 2}For $i=k+l+1,...,n_{A}$, $\mathbf{v}_{A;i}$ is an
eigenvector of $\mathbf{U}_{A}\mathbf{U}_{A}^{\dag }+\mathbf{U}_{B}\mathbf{U}%
_{B}^{\dag }$ corresponding to $\lambda _{i}=1$, and is orthogonal to $%
\mathcal{C}(\mathbf{H}_{BR})$; for $i=k+l+1,...,n_{B}$, $\mathbf{v}_{B;i}$
is an eigenvector corresponding to $\lambda _{i}=1$, and is orthogonal to $%
\mathcal{C}(\mathbf{H}_{AR})$.
\end{coro}

\begin{rema}
The above corollaries show that the eigenvalue $\lambda _{i}$ is an
indicator of the degree of orthogonality of the $i$th direction pair. In
particular, $\lambda _{i}\approx 2$ means that the two channel directions
are close to parallel; and $\lambda _{i}\approx 1$ means that the two channel
directions are close to orthogonal.
\end{rema}

\begin{rema}
\label{Remark 4}Before leaving this subsection, we emphasize that the joint
channel decomposition in Theorem \ref{Theorem 1} is general for any sizes of
$\mathbf{H}_{AR}$ and $\mathbf{H}_{BR}$. Particularly, if $n_{m}\geq
n_{R},m\in \left\{ A,B\right\} $, then $k=n_{R}$ and $l=0$, implying that
all the eigenvalues $\{\lambda _{i}\}$ are valued at $2$. In this case, $%
\mathbf{H}_{AR}$ and $\mathbf{H}_{BR}$ span the same columnspace. Channel
alignment techniques have been proposed in \cite{HJYangIT11}-\cite%
{KhinaISIT11} for efficient implementation of PNC. In what follows, we are
mainly interested in the case of $n_{A},n_{B}<n_{R}$, i.e., there exist $%
\{\lambda _{i}\}$ valued between, but not including, $1$ and $2$.
\end{rema}

\subsection{Space-Division Approach for MIMO Two-Way Relaying}

Based on the joint channel decomposition in Theorem \ref{Theorem 1}, we now
propose a new space-division approach for MIMO two-way relaying. The main
idea is to divide the overall signal space $\mathcal{C}([\mathbf{H}_{AR}$ $%
\mathbf{H}_{BR}]\mathbf{)}=\mathcal{C}(\mathbf{U)}$ into two orthogonal
subspaces: 1) $\mathcal{S}^{PNC}$, in which the channel direction pairs ($%
\mathbf{v}_{A;i},\mathbf{v}_{B;i}$) are parallel or close to parallel, for
carrying out PNC; 2) $\mathcal{S}^{CD}$ for carrying out the
complete-decoding strategy. Let $l%
{\acute{}}%
$ be an arbitrary integer between 0 and $l$. Recall that the channel
direction pairs are ordered by the degree of orthogonality as in Corollary %
\ref{Corollary 1}. Therefore, the first $k+l%
{\acute{}}%
$ direction pairs have lower degree of orthogonality compared to the
remaining pairs. Thus, we allocate the first $k+l%
{\acute{}}%
$ direction pairs to form a basis of $\mathcal{S}^{PNC}$. The remaining
channel directions give a basis of $\mathcal{S}^{CD}$. In this section, we
assume that $l%
{\acute{}}%
$ is given. The details on the optimization of $l%
{\acute{}}%
$ will be discussed later in Sections V and VI.

\subsubsection{Space-Division Operation}

Let the RQ decomposition of $\mathbf{G}_{m}$ be%
\begin{equation}
\mathbf{G}_{m}=\mathbf{R}_{m}\mathbf{T}_{m}^{\dagger },\text{ }m\in \left\{
A,B\right\}
\end{equation}%
where $\mathbf{R}_{m}\in $ $\mathcal{%
\mathbb{C}
}^{n_{m}\times n_{m}}$ is an upper-triangular matrix given by%
\begin{equation}
\mathbf{R}_{m}=\left[
\begin{array}{cccc}
r_{m;1,1} & r_{m;1,2} & \cdots & r_{m;1,n_{m}} \\
0 & r_{m;2,2} & \cdots & r_{m;2,n_{m}} \\
\vdots & \vdots & \ddots & \vdots \\
0 & \cdots & 0 & r_{m;n_{m},n_{m}}%
\end{array}%
\right] ,
\end{equation}%
and $\mathbf{T}_{m}\in $ $\mathcal{%
\mathbb{C}
}^{n_{m}\times n_{m}}$ is unitary. Together with (\ref{Hm}), the channel
matrices\textit{\ }can be jointly decomposed as%
\begin{equation}
\mathbf{H}_{m}=\mathbf{UD}_{m}\mathbf{R}_{m}\mathbf{T}_{m}^{\dag },\mathbf{\
}m\in \left\{ A,B\right\} .  \label{Decomp}
\end{equation}
Then, the received signal at the relay, after
left-multiplying $\mathbf{U}^{\dag }$, can be represented as%
\begin{equation}
\mathbf{Y}_{R}^{\prime }=\mathbf{U}^{\dag }\mathbf{Y}_{R}=\mathbf{D}_{A}%
\mathbf{R}_{A}\mathbf{X}_{A}^{\prime }+\mathbf{D}_{B}\mathbf{R}_{B}\mathbf{X}%
_{B}^{\prime }+\mathbf{Z}_{R}^{\prime },  \label{33}
\end{equation}%
where $\mathbf{X}_{m}^{\prime }=\mathbf{T}_{m}^{\dag }\mathbf{X}_{m},m\in
\left\{ A,B\right\} $, and $\mathbf{Z}_{R}^{\prime }=\mathbf{U}^{\dag }%
\mathbf{Z}_{R}$ with i.i.d. elements $\sim \mathcal{N}_{c}(0,N_{0})$.

We partition $\mathbf{R}_{m}$ and $\mathbf{D}_{m}$ as%
\begin{equation}
\mathbf{R}_{m}=\left[
\begin{array}{cc}
\mathbf{R}_{m;1,1} & \mathbf{R}_{m;1,2} \\
\mathbf{0} & \mathbf{R}_{m;2,2}%
\end{array}%
\right] \text{ and }\mathbf{D}_{m}=\left[
\begin{array}{cc}
\mathbf{D}_{m;1,1} & \mathbf{0} \\
\mathbf{0} & \mathbf{D}_{m;2,2}%
\end{array}%
\right] ,m\in \left\{ A,B\right\}  \label{Dm3}
\end{equation}%
where $\mathbf{R}_{m;1,1}\in
\mathbb{C}
^{(k+l%
{\acute{}}%
)\times (k+l%
{\acute{}}%
)}$ and $\mathbf{R}_{m;2,2}\in
\mathbb{C}
^{(l-l%
{\acute{}}%
+p_{m})\times (l-l%
{\acute{}}%
+p_{m})},$ $m\in \left\{ A,B\right\} $, are upper triangular matrices, and $%
\mathbf{D}_{m;1,1}\in
\mathbb{C}
^{(k+2l%
{\acute{}}%
)\times (k+l)}$ and $\mathbf{D}_{m;2,2}\in
\mathbb{C}
^{(n_{R}-k-2l%
{\acute{}}%
)\times (l-l%
{\acute{}}%
+p_{m})},m\in \left\{ A,B\right\} $, are block-diagonal matrices. Then, (\ref%
{33}) can be written as%
\begin{equation}
\left[
\begin{array}{c}
\mathbf{Y}_{R}^{PNC} \\
\mathbf{Y}_{R}^{CD}%
\end{array}%
\right] =\dsum\limits_{m\in \left\{ A,B\right\} }\left[
\begin{array}{cc}
\mathbf{D}_{m;1,1}\mathbf{R}_{m;1,1} & \mathbf{D}_{m;1,1}\mathbf{R}_{m;1,2}
\\
\mathbf{0} & \mathbf{D}_{m;2,2}\mathbf{R}_{m;2,2}%
\end{array}%
\right] \left[
\begin{array}{c}
\mathbf{X}_{m}^{PNC} \\
\mathbf{X}_{m}^{CD}%
\end{array}%
\right] +\left[
\begin{array}{c}
\mathbf{Z}_{R}^{PNC} \\
\mathbf{Z}_{R}^{CD}%
\end{array}%
\right] ,  \label{SS}
\end{equation}%
where $\mathbf{Y}_{R}^{\prime }$, $\mathbf{X}_{m}^{\prime }$, and $\mathbf{Z}%
_{R}^{\prime }$ are correspondingly partitioned as%
\begin{equation}
\mathbf{Y}_{R}^{\prime }=\left[
\begin{array}{c}
\mathbf{Y}_{R}^{PNC} \\
\mathbf{Y}_{R}^{CD}%
\end{array}%
\right] \text{, }\mathbf{X}_{m}^{\prime }=\left[
\begin{array}{c}
\mathbf{X}_{m}^{PNC} \\
\mathbf{X}_{m}^{CD}%
\end{array}%
\right] ,\text{ and }\mathbf{Z}_{R}^{\prime }=\left[
\begin{array}{c}
\mathbf{Z}_{R}^{PNC} \\
\mathbf{Z}_{R}^{CD}%
\end{array}%
\right] .
\end{equation}%
Here, the superscript \textquotedblleft PNC\textquotedblright\ (or
\textquotedblleft CD\textquotedblright ) represents the PNC (or
complete-decoding) strategy.

Based on the signal model in (\ref{SS}), the proposed space-division based
relaying strategy is described as follows. At user $m$, two groups of
spatial streams are generated: one group, referred to as the \textit{%
complete-decoding spatial streams}, form the codeword matrix $\mathbf{X}%
_{m}^{CD}$; and the other group, referred to as the \textit{PNC spatial
streams}, form the codeword matrix $\mathbf{X}_{m}^{PNC}$, $m\in \left\{
A,B\right\} $.

\subsubsection{Complete-Decoding Spatial Streams}

Due to the block triangular structure of the channel matrices in (\ref{SS}),
the relay can completely decode the spatial streams $\mathbf{X}_{A}^{CD}$
and $\mathbf{X}_{B}^{CD}$ free of interference from the PNC spatial streams.
Specifically, the relay completely decodes both $\mathbf{X}_{A}^{CD}$ and $%
\mathbf{X}_{B}^{CD}$ based on%
\begin{equation}
\mathbf{Y}_{R}^{CD}=\dsum\limits_{m\in \left\{ A,B\right\} }\mathbf{D}%
_{m;2,2}\mathbf{R}_{m;2,2}\mathbf{X}_{m}^{CD}+\mathbf{Z}_{R}^{CD}.
\label{S_CD}
\end{equation}%
Then, $\mathbf{X}_{A}^{CD}$ and $\mathbf{X}_{B}^{CD}$ are canceled from the
received signal in (\ref{SS}).

\subsubsection{PNC Spatial Streams}

After the cancelation of $\mathbf{X}_{A}^{CD}$ and $\mathbf{X}_{B}^{CD}$, the system model for the PNC spatial streams is
given by%
\begin{equation}
\mathbf{Y}_{R}^{PNC}=\dsum\limits_{m\in \left\{ A,B\right\} }\mathbf{D}%
_{m;1,1}\mathbf{R}_{m;1,1}\mathbf{X}_{m}^{PNC}+\mathbf{Z}_{R}^{PNC}.
\end{equation}%
From (\ref{Dm}), the first $k$ columns of $\mathbf{D}_{A;1,1}$ and $\mathbf{D%
}_{B;1,1}$ are identical; however, for $i=k+1,...,k+l%
{\acute{}}%
$, the $i$th columns of $\mathbf{D}_{A;1,1}$ and $\mathbf{D}_{B;1,1}$ are
not. Following Section III, we project each column pair of $%
\mathbf{D}_{A;1,1}$ and $\mathbf{D}_{B;1,1}$ onto a common direction, so as
to facilitate PNC.

By inspection, the only difference between the $i$th columns of $\mathbf{D}%
_{A;1,1}$ and $\mathbf{D}_{B;1,1}$ is given by the 2-by-1 vectors $\mathbf{e}%
_{A;i}$ and $\mathbf{e}_{B;i}$, for $i=k+1,...,k+l%
{\acute{}}%
$. Without loss of generality, denote by $\mathbf{p}_{i}$ a 2-by-1 unit
vector representing the projection direction of $\mathbf{e}_{A;i}$ and $%
\mathbf{e}_{B;i}$. The choice of $\mathbf{p}_{i}$ is similar to that
described in Section III and will be detailed in the next section.

Now the projection process can be described in a matrix form as follows.
Define the projection matrix%
\begin{equation}
\mathbf{P}=\left[
\begin{array}{cccc}
\mathbf{I}_{k} & \mathbf{0} & \cdots & \mathbf{0} \\
\mathbf{0} & \mathbf{p}_{k+1} & \ddots & \vdots \\
\vdots & \ddots & \ddots & \mathbf{0} \\
\mathbf{0} & \cdots & \mathbf{0} & \mathbf{p}_{k+l%
{\acute{}}%
}%
\end{array}%
\right] \in
\mathbb{C}
^{(k+2l%
{\acute{}}%
)\times (k+l%
{\acute{}}%
)}\text{. }
\end{equation}%
After the projection, the resulting signal model is given by%
\begin{equation}
\widetilde{\mathbf{Y}}_{R}^{PNC}=\mathbf{P}^{T}\mathbf{Y}_{R}^{PNC}=\dsum%
\limits_{m\in \left\{ A,B\right\} }\widetilde{\mathbf{H}}_{m}^{PNC}\mathbf{X}%
_{m}^{PNC}+\widetilde{\mathbf{Z}}_{R}^{PNC}  \label{S_relay5}
\end{equation}%
where $\widetilde{\mathbf{H}}_{m}^{PNC}=\mathbf{P}^{T}\mathbf{D}_{m;1,1}%
\mathbf{R}_{m;1,1}=\mathbf{\tilde{D}}_{m;1,1}\mathbf{R}_{m;1,1}\in
\mathbb{C}
^{(k+l%
{\acute{}}%
)\times (k+l%
{\acute{}}%
)}$, with%
\begin{equation}
\mathbf{\tilde{D}}_{m;1,1}=\text{diag}\left\{ 1,...,1,\mathbf{p}_{k+1}^{T}%
\mathbf{e}_{m;k+1},...,\mathbf{p}_{k+l%
{\acute{}}%
}^{T}\mathbf{e}_{m;k+l%
{\acute{}}%
}\right\} ,m\in \left\{ A,B\right\} ,  \label{ddd}
\end{equation}%
and $\widetilde{\mathbf{Z}}_{R}^{PNC}=\mathbf{P}^{T}\mathbf{Z}_{R}^{PNC}$ $%
\in
\mathbb{C}
^{(n_{R}-l%
{\acute{}}%
)\times 1}$ with the entries being i.i.d. random variables $\sim \mathcal{N}%
_{c}(0,N_{0}\dot{)}$. Note that the equivalent channel matrices $\widetilde{%
\mathbf{H}}_{A}^{PNC}$ and $\widetilde{\mathbf{H}}_{B}^{PNC}$ are $(k+l%
{\acute{}}%
)$-by-$(k+l%
{\acute{}}%
)$ square matrices. For such an equivalent MIMO TWRC, efficient techniques
can be employed to align the signal directions of the two user into a same
set of $k+l%
{\acute{}}%
$ directions. This provides a platform to carry out $k$ PNC streams.

So far, we have presented the signal processing techniques used in the proposed space-division scheme to
manipulate the uplink channel. The
encoding and decoding details of the overall scheme will be described in
the next section.

\section{An Achievable Rate-Region of MIMO TWRC}

In this section, we derive an achievable rate-pair of the proposed
space-division based PNC scheme. Based on that, we optimize the
system parameters to determine the achievable rate-region.

\subsection{Achievable Rate-Pairs}

\subsubsection{Complete-Decoding Spatial Streams}

The equivalent channel model seen by the complete-decoding
spatial streams is given in (\ref{S_CD}), with the equivalent channel
matrices given by $\mathbf{D}_{m,2,2}\mathbf{R}_{m;2,2}$, $m\in \left\{
A,B\right\} $.

The signal model in (\ref{S_CD}) is a standard MIMO MAC channel. Let $%
\mathbf{Q}_{m}^{CD}=\frac{1}{T}E\left[ \mathbf{X}_{m}^{CD}(\mathbf{X}%
_{m}^{CD})^{\dag }\right] $ be the input covariance matrix of the
complete-decoding spatial steams of user $m$. Then, the achievable rate-pair
of the complete-decoding spatial streams satisfies \cite{Cover91}%
\begin{subequations}\label{Rate CD}
\begin{eqnarray}
R_{A}^{CD}+R_{B}^{CD} &\leq &\frac{1}{2}\log \left\vert \mathbf{I}+\frac{1}{%
N_{0}}\dsum\limits_{m\in \left\{ A,B\right\} }\mathbf{D}_{m;2,2}\mathbf{R}%
_{m;2,2}\mathbf{Q}_{m}^{CD}\mathbf{R}_{m;2,2}^{\dag }\mathbf{D}%
_{m;2,2}^{\dag }\right\vert   \\
R_{m}^{CD} &\leq &\frac{1}{2}\log \left\vert \mathbf{I}+\frac{1}{N_{0}}%
\mathbf{D}_{m;2,2}\mathbf{R}_{m;2,2}\mathbf{Q}_{m}^{CD}\mathbf{R}%
_{m;2,2}^{\dag }\mathbf{D}_{m;2,2}^{\dag }\right\vert ,m\in \left\{
A,B\right\} .
\end{eqnarray}
\end{subequations}

\subsubsection{PNC Spatial Streams}

The equivalent channel seen by the PNC streams is given in (\ref{S_relay5}).
Recall that $\widetilde{\mathbf{H}}_{m}^{PNC}$ is a ($k+l%
{\acute{}}%
$)-by-($k+l%
{\acute{}}%
$) square matrix, and the efficient design of PNC for this case has been
discussed in \cite{HJYangIT11}-\cite{KhinaISIT11}. Here, we follow the
GSVD-based approach in \cite{HJYangIT11}, as briefly described below.

Applying the generalized singular-value decomposition (GSVD) \cite{GolubMC97}
to $\widetilde{\mathbf{H}}_{m}^{PNC}$, we obtain%
\begin{equation}
\widetilde{\mathbf{H}}_{m}^{PNC}=\mathbf{B\Sigma }_{m}\mathbf{T}%
_{m}^{\dagger },m\in \left\{ A,B\right\} ,
\end{equation}%
where $\mathbf{B}\in
\mathbb{C}
^{(k+l%
{\acute{}}%
)\times (k+l%
{\acute{}}%
)}$ is a nonsingular matrix, $\mathbf{T}_{m}\in
\mathbb{C}
^{(k+l%
{\acute{}}%
)\times (k+l%
{\acute{}}%
)}$ is an orthogonal matrix, $m\in \left\{ A,B\right\} $, and $\mathbf{%
\Sigma }_{m}\in
\mathbb{C}
^{(k+l%
{\acute{}}%
)\times (k+l%
{\acute{}}%
)}$ is a diagonal matrix with the $i$th diagonal element denoted by $\sigma
_{m;i}$. We further take the QR decomposition to the matrix $\mathbf{B}$,
yielding%
\begin{equation}
\widetilde{\mathbf{H}}_{m}^{PNC}=\mathbf{Q\tilde{R}\Sigma }_{m}\mathbf{T}%
_{m}^{\dagger },m\in \left\{ A,B\right\} ,  \label{Decom}
\end{equation}%
where $\mathbf{\tilde{R}}\in
\mathbb{C}
^{(k+l%
{\acute{}}%
)\times (k+l%
{\acute{}}%
)}$ is an upper triangular matrix. The transmit signal $\mathbf{X}_{m}^{PNC}$
in (\ref{S_relay5}) is designed as%
\begin{equation}
\mathbf{X}_{m}^{PNC}=\mathbf{T}_{m}\mathbf{\Psi }_{m}^{1/2}\mathbf{S}%
_{m}^{PNC},m\in \left\{ A,B\right\} ,  \label{SigModel}
\end{equation}%
where $\mathbf{\Psi }_{m}^{1/2}=$ diag$\left\{ \sqrt{\psi _{m;1}},\sqrt{\psi
_{m;2}},\cdots ,\sqrt{\psi _{m;k+l%
{\acute{}}%
}}\right\} $ is a diagonal matrix with $\psi _{m;i}\geq 0,i=1,2,...,k+l%
{\acute{}}%
$, and $\mathbf{S}_{m}$ $\in
\mathbb{C}
^{(k+l%
{\acute{}}%
)\times T}$ is the signal matrix with each element independent and
identically drawn from $\mathcal{N}_{c}(0,1)$.

Let $R_{m}^{PNC}$ be the total rate of the PNC spatial streams of user $m$.
From Theorem 1 in \cite{HJYangIT11}, the achievable rate-pair is given by%
\begin{equation}
R_{m}^{PNC}=\sum_{i=1}^{k+l%
{\acute{}}%
}\frac{1}{2}\left[ \log \left( \frac{I(i)\sigma _{m;i}^{2}\psi _{m;i}^{{}}}{%
\sigma _{A;i}^{2}\psi _{A;i}^{{}}+\sigma _{B;i}^{2}\psi _{B;i}^{{}}}+\frac{%
\tilde{r}_{i,i}^{2}\sigma _{m;i}^{2}\psi _{m;i}^{{}}}{N_{0}}\right) \right]
^{+},m\in \left\{ A,B\right\}  \label{Rate PNC}
\end{equation}%
where $I(i)$ is the indicator function with $I(i)=1$ for $i=1$ and $I(i)=0$
for $i\neq 1$.

\subsubsection{The Overall Scheme}

We now consider the overall achievable rate-pair of the proposed
space-division based PNC scheme. Before going into details, we
note that the power constraint of user $m$, i.e., $\frac{1}{T}E\left[
\left\Vert \mathbf{X}_{m}\right\Vert _{F}^{2}\right] \leq P_{m},$ $m\in
\left\{ A,B\right\} $, can be equivalently expressed as%
\begin{equation}
\text{tr}\{\mathbf{Q}_{m}^{CD}\}+\sum_{i=1}^{k+l%
{\acute{}}%
}\psi _{m;i}^{{}}\leq P_{m},m\in \left\{ A,B\right\} .  \label{PC1}
\end{equation}%
We are now ready to present the following theorem on the achievable rates of
the proposed scheme.

\begin{theo}
For given $\mathbf{Q}_{m}^{CD}$, $\mathbf{\Psi }_{m}^{{}}$, and $\mathbf{Q}%
_{R}$ satisfying (\ref{PC1}) and tr$(\mathbf{Q}_{R})\leq P_{R}$, a rate-pair
$(R_{A},R_{B})$ for the MIMO\ TWRC is achievable if%
\begin{equation}
R_{m}\leq \min \{R_{m}^{CD}+R_{m}^{PNC},R_{m}^{DL}\},m\in \left\{
A,B\right\} ,  \label{49}
\end{equation}%
where $R_{A}^{CD}$ and $R_{B}^{CD}$ satisfy (\ref{Rate CD}), $R_{m}^{PNC}$
is given by (\ref{Rate PNC}), and $R_{m}^{DL}$ is given by (\ref{Outer Bound}%
).
\end{theo}

\begin{proof}
Here we provide a sketch of proof. The overall encoding and decoding process
for the proposed scheme is described as follows. The messages of the user $m$
are doubly indexed as $(W_{m}^{CD},W_{m}^{PNC})$, with $W_{m}^{CD}\in
\{1,2,...,2^{2TR_{m}^{CD}}\}$ for the complete-decoding spatial streams, and
$W_{m}^{PNC}\in \{1,2,...,2^{2TR_{m}^{PNC}}\}$ for the PNC spatial streams.
Each $W_{m}^{CD}$ is one-to-one mapped to $\mathbf{X}_{m}^{CD}$ in (\ref{SS}%
), and each $W_{m}^{PNC}$ is one-to-one mapped to $\mathbf{S}_{m}^{PNC}$ in (%
\ref{SigModel}). In the uplink phase, $\mathbf{X}_{m}^{CD}$ and $\mathbf{X}%
_{m}^{PNC}=\mathbf{T}_{m}\mathbf{\Psi }_{m}^{1/2}\mathbf{S}_{m}^{PNC}$ are
transmitted via the channel in (\ref{SS}), with the transmit power
constrained by (\ref{PC1}).

Upon receiving $\mathbf{Y}_{R}$, the relay first completely decode $\mathbf{X%
}_{A}^{CD}$ and $\mathbf{X}_{B}^{CD}$ based on $\mathbf{Y}_{R}^{CD}$ in (\ref%
{S_CD}), with the achievable rate-pair given in (\ref{Rate CD}). The decoded
$\mathbf{X}_{A}^{CD}$ and $\mathbf{X}_{B}^{CD}$ are subtracted from $\mathbf{%
Y}_{R}^{PNC}$. Let $(\mathbf{s}_{m,i}^{PNC})^{T}$ be the $i$th row of $%
\mathbf{S}_{m}^{PNC}$. Then, the network-coded PNC spatial streams, i.e., $%
\tilde{r}_{i,i}^{{}}\sigma _{A;i}^{{}}\psi _{A;i}^{1/2}\mathbf{s}%
_{A,i}^{PNC}+\tilde{r}_{i,i}^{{}}\sigma _{B;i}^{{}}\psi _{B;i}^{1/2}\mathbf{s%
}_{B,i}^{PNC}$, $i=k+l{\acute{}},k+l{\acute{}}-1,...,1$, are successively recovered and canceled from $\widetilde{\mathbf{Y%
}}_{R}^{PNC}$ in (\ref{S_relay5}), with the achievable rate-pair given in (%
\ref{Rate PNC}). The decoded messages from the complete-decoding streams, together with the network-coded messages from the PNC streams, are then jointly
encoded. The new codeword is forwarded to the two users in the
downlink phase, with the transmit power constrained by tr$(\mathbf{Q}%
_{R})\leq P_{R}$. Following the discussions in \cite{HJYangIT11}-\cite%
{KhinaISIT11}, the achievable rate-pair of the downlink phase is given by ($%
R_{A}^{DL},R_{B}^{DL}$) in (\ref{Outer Bound}). This completes the proof.
\end{proof}

\subsection{Determining Achievable Rate-Region}

Now we consider determining the boundary of the achievable rate-region. From
(\ref{49}), the downlink achievable rates are the same as the capacity upper
bound in (\ref{Outer Bound}). Here, we focus on the uplink rate-region.

The boundary of the uplink rate-region can be determined by solving the
following weighted-sum-rate maximization problem:%
\begin{subequations}\label{WSR}
\begin{eqnarray}
\text{maximize} &&\dsum\limits_{m\in \left\{ A,B\right\} }w_{m}\left(
R_{m}^{CD}+R_{m}^{PNC}\right)   \\
\text{subject to} &&\dsum\limits_{j=1}^{k+l%
{\acute{}}%
}\psi _{m;j}+\text{tr}\{\mathbf{Q}_{m}^{CD}\}\leq P_{m},\mathbf{Q}%
_{m}^{CD}\succeq \mathbf{0},\psi _{m;i}\geq 0,\text{ for }i=1,...,k+l%
{\acute{}}%
. \\
&&R_{m}^{PNC}=\sum_{i=1}^{k+l%
{\acute{}}%
}\frac{1}{2}\left[ \log \left( \frac{I(i)\sigma _{m;i}^{2}\psi _{m;i}^{{}}}{%
\sigma _{A;i}^{2}\psi _{A;i}^{{}}+\sigma _{B;i}^{2}\psi _{B;i}^{{}}}+\frac{%
\tilde{r}_{i,i}^{2}\sigma _{m;i}^{2}\psi _{m;i}^{{}}}{N_{0}}\right) \right]
^{+}, \\
&&R_{A}^{CD}+R_{B}^{CD}\leq \frac{1}{2}\log \left\vert \mathbf{I}+\frac{1}{%
N_{0}}\dsum\limits_{m\in \left\{ A,B\right\} }\mathbf{D}_{m;2,2}\mathbf{R}%
_{m;2,2}\mathbf{Q}_{m}^{CD}\mathbf{R}_{m;2,2}^{\dag }\mathbf{D}%
_{m;2,2}^{\dag }\right\vert , \\
&&R_{m}^{CD}\leq \frac{1}{2}\log \left\vert \mathbf{I}+\frac{1}{N_{0}}%
\mathbf{D}_{m;2,2}\mathbf{R}_{m;2,2}\mathbf{Q}_{m}^{CD}\mathbf{R}%
_{m;2,2}^{\dag }\mathbf{D}_{m;2,2}^{\dag }\right\vert ,m\in \left\{
A,B\right\} .
\end{eqnarray}%
\end{subequations}
The above problem involves the optimization of $l%
{\acute{}}%
,\{\mathbf{p}_{i}\}_{i=k+1}^{k+l%
{\acute{}}%
},\mathbf{Q}_{A}^{CD},\mathbf{Q}_{B}^{CD},\left\{ \psi _{A;i}\right\}
_{i=1}^{k+l%
{\acute{}}%
}$, and $\left\{ \psi _{B;i}\right\} _{i=1}^{k+l%
{\acute{}}%
}$, as detailed below.

\subsubsection{Determining the Projection Directions}

The optimization of the projection directions $\{\mathbf{p}%
_{i}\}_{i=k+1}^{k+l%
{\acute{}}%
}$ to maximize the weighted sum-rate is in general difficult to
solve. To simplify the problem, we consider the high SNR regime, with the
weighted sum-rate given by%
\begin{eqnarray}
&&w_{A}R_{A}^{PNC}+w_{B}R_{B}^{PNC}\overset{(a)}{\approx }\frac{1}{2}%
\dsum\limits_{m\in \left\{ A,B\right\} }w_{m}\left( \sum_{i=2}^{k+l%
{\acute{}}%
}\log \left( \frac{\tilde{r}_{i,i}^{2}\sigma _{m;i}^{2}\psi _{m;i}^{{}}}{%
N_{0}}\right) \right)  \notag \\
&&\overset{(b)}{=}\frac{1}{2}\dsum\limits_{m\in \left\{ A,B\right\}
}w_{m}\log \left\vert \frac{P_{m}}{N_{0}n_{m}}\mathbf{\tilde{R}\Sigma }_{m}%
\mathbf{\Sigma }_{m}^{\dag }\mathbf{\tilde{R}}^{\dag }\right\vert  \notag \\
&&\overset{(c)}{=}\frac{1}{2}\dsum\limits_{m\in \left\{ A,B\right\}
}w_{m}\log \left\vert \frac{P_{m}}{N_{0}n_{m}}\mathbf{\tilde{D}}_{m;1,1}%
\mathbf{R}_{m;1,1}\mathbf{R}_{m;1,1}^{\dag }\mathbf{\tilde{D}}_{m;1,1}^{\dag
}\right\vert  \notag \\
&&\overset{(d)}{=}\frac{1}{2}\dsum\limits_{m\in \left\{ A,B\right\} }\left(
w_{m}\log \left\vert \frac{P_{m}}{N_{0}n_{m}}\mathbf{R}_{m;1,1}\mathbf{R}%
_{m;1,1}^{\dag }\right\vert +w_{m}\log \left\vert \mathbf{\tilde{D}}_{m;1,1}%
\mathbf{\tilde{D}}_{m;1,1}^{\dag }\right\vert \right)  \label{WSR1}
\end{eqnarray}%
where step (\textit{a}) follows from substituting (\ref{Rate PNC}), step (%
\textit{b}) from the facts that $\mathbf{\tilde{R}}$ is upper-triangular and
that equal power allocation is asymptotically optimal (i.e., $\psi
_{m;i}^{{}}=\frac{P_{m}}{n_{m}}$), step (\textit{c}) by noting $\mathbf{%
\tilde{D}}_{m;1,1}\mathbf{R}_{m;1,1}=\widetilde{\mathbf{H}}_{m}^{PNC}=%
\mathbf{Q\tilde{R}\Sigma }_{m}\mathbf{T}_{m}^{\dagger }$ (cf., (\ref{S_relay5}) and (\ref{Decom})), and step (d) by
utilizing $\left\vert \mathbf{I}+\mathbf{AB}\right\vert =\left\vert \mathbf{I}+%
\mathbf{BA}\right\vert $. In the above, $\log \left\vert \mathbf{\tilde{D}}%
_{m;1,1}\mathbf{\tilde{D}}_{m;1,1}^{\dag }\right\vert $ is the only term
related to $\mathbf{p}_{i}$. Recall from (\ref{ddd}) that $\mathbf{\tilde{D}}_{m;1,1}=%
\mathbf{P}^{T}\mathbf{D}_{m;1,1}$ with $\mathbf{D}_{m;1,1}$ being the
principle submatrix of $\mathbf{D}_{m}$ in (\ref{Dm3}). Thus, the weighted
sum-rate maximization problem over $\{\mathbf{p}_{i}\}_{i=k+1}^{k+l%
{\acute{}}%
}$ can be decoupled into $l%
{\acute{}}%
$ independent subproblems as%
\begin{equation}
\max_{||\mathbf{p}_{i}||=1}w_{A}\log \left( \left\vert \mathbf{p}_{i}^{\dag }%
\mathbf{e}_{A;i}\right\vert ^{2}\right) +w_{B}\log \left( \left\vert \mathbf{%
p}_{i}^{\dag }\mathbf{e}_{B;i}\right\vert ^{2}\right) ,\text{ for }%
i=k+1,...,k+l%
{\acute{}}%
,
\end{equation}%
where $\mathbf{e}_{A;i}$ and $\mathbf{e}_{B;i}$ are given in (\ref{Direction}%
). From (\ref{d}) and the discussions therein, the optimal $\mathbf{p}_{i}$ to
maximize the weighted sum-rate is a real vector given by%
\begin{equation}
\mathbf{p}_{i}=\gamma _{i}\left( \mathbf{e}_{A;i}+\beta _{i}\mathbf{e}%
_{B;i}\right) ,\text{ for }i=k+1,...,k+l%
{\acute{}}%
,  \label{dd}
\end{equation}%
where%
\begin{equation}
\beta _{i}=\frac{1}{2}\left( \sqrt{\left( \lambda _{i}-1\right) ^{2}\left( 1-%
\frac{w_{B}}{w_{A}}\right) ^{2}+4\frac{w_{B}}{w_{A}}}-\left( \lambda
_{i}-1\right) \left( 1-\frac{w_{B}}{w_{A}}\right) \right) ,
\end{equation}%
and $\gamma _{i}$ is a scaling factor to ensure $||\mathbf{p}_{i}||$ $=1$.

\subsubsection{Determining $\mathbf{Q}_{A}^{CD}$ and $\mathbf{Q}_{B}^{CD}$}

Given $\{\mathbf{p}_{i}\}$ in (\ref{dd}), the optimization problem in (\ref%
{WSR}) can be decoupled into two separate problems by predetermining the
power allocated to the two signal subspaces. Let $P_{m}^{CD}$ be the power
of user $m$ used for the complete-decoding spatial streams. Then, the power
for the PNC streams is given by $P_{m}^{PNC}=P_{m}^{{}}-P_{m}^{CD},m\in
\left\{ A,B\right\} $. For given $P_{A}^{CD}$ and $P_{B}^{CD}$, the optimal $%
\mathbf{Q}_{A}^{CD}$ and $\mathbf{Q}_{B}^{CD}$ to (\ref{WSR}) can be found
by solving the following problem:
\begin{subequations}
\begin{eqnarray}
\text{maximize} &&w_{A}R_{A}^{CD}+w_{A}R_{A}^{CD} \\
\text{subject to} &&\text{tr}\{\mathbf{Q}_{m}^{CD}\}\leq P_{m}^{CD},\mathbf{Q%
}_{m}^{CD}\succeq \mathbf{0}, \\
&&R_{A}^{CD}+R_{B}^{CD}\leq \frac{1}{2}\log \left\vert \mathbf{I}+\frac{1}{%
N_{0}}\dsum\limits_{m\in \left\{ A,B\right\} }\mathbf{D}_{m,2,2}\mathbf{R}%
_{m;2,2}\mathbf{Q}_{m}^{CD}\mathbf{R}_{m;2,2}^{\dag }\mathbf{D}%
_{m,2,2}^{\dag }\right\vert \\
&&R_{m}^{CD}\leq \frac{1}{2}\log \left\vert \mathbf{I}+\frac{1}{N_{0}}%
\mathbf{D}_{m,2,2}\mathbf{R}_{m;2,2}\mathbf{Q}_{m}^{CD}\mathbf{R}%
_{m;2,2}^{\dag }\mathbf{D}_{m,2,2}^{\dag }\right\vert ,m\in \left\{
A,B\right\} .
\end{eqnarray}%
\end{subequations}
The above is a standard weighted sum-rate maximization problem for a MIMO
multiple-access channel with two users \cite{YuIT04}. This problem is
convex, and the optimal solution can be numerically obtained using convex
programming tools \cite{Boyd04}.

\subsubsection{Determining Power Allocation for PNC Streams}

Now we consider the optimization of $\left\{ \psi _{A;i}\right\} _{i=1}^{k+l%
{\acute{}}%
}$ and $\left\{ \psi _{B;i}\right\} _{i=1}^{k+l%
{\acute{}}%
}$. Given $P_{A}^{PNC}$ and $P_{B}^{PNC}$, the optimal $\left\{ \psi
_{A;i}\right\} _{i=1}^{k+l%
{\acute{}}%
}$ and $\left\{ \psi _{B;i}\right\} _{i=1}^{k+l%
{\acute{}}%
}$ can be determined by solving
\begin{subequations}
\begin{eqnarray}
\text{maximize} &&\dsum\limits_{m\in \left\{ A,B\right\} }w_{m}\left(
\sum_{i=1}^{k+l%
{\acute{}}%
}\frac{1}{2}\left[ \log \left( \frac{I(i)\sigma _{m;i}^{2}\psi _{m;i}^{{}}}{%
\sigma _{A;i}^{2}\psi _{A;i}^{{}}+\sigma _{B;i}^{2}\psi _{B;i}^{{}}}+\frac{%
\tilde{r}_{i,i}^{2}\sigma _{m;i}^{2}\psi _{m;i}^{{}}}{N_{0}}\right) \right]
^{+}\right) \\
\text{subject to} &&\dsum\limits_{i=1}^{k+l%
{\acute{}}%
}\psi _{m;i}\leq P_{m}^{PNC},\psi _{m;i}\geq 0,\text{ for }i=1,...,k+l%
{\acute{}}%
.
\end{eqnarray}%
\end{subequations}
A similar problem has been considered in \cite{HJYangIT11}, and the optimal
solution can be obtained by solving the Karush-Kuhn-Tuchker (KKT)
conditions. We omit details here for simplicity.

Base on the above discussions, the weighted sum-rate problem in (\ref{WSR})
is numerically solvable given the values of $l%
{\acute{}}%
,P_{m}^{CD}$ and $P_{m}^{PNC}$, $m\in \left\{ A,B\right\} $. The optimal $l%
{\acute{}}%
,P_{m}^{CD}$ and $P_{m}^{PNC}$, $m\in \left\{ A,B\right\} $ can be found
using the exhaustive search. The complexity involved is not significant by
noting $P_{m}^{CD}+P_{m}^{PNC}=P_{m},m\in \left\{ A,B\right\} $ and the fact
that $l%
{\acute{}}%
$ is an integer between $0$ and $l$.

\section{Asymptotic Sum-Rate Analysis}

In the preceding section, we have shown the achievable rates of the proposed
space-division based network-coding strategy for MIMO TWRCs. In general, it
is difficult to represent the achievable rate of the optimized
space-division based scheme in a closed-form. Thus, it is not easy to
evaluate the gap between the achievable rate of the proposed scheme and the
capacity upper bound of the MIMO TWRC. In this section, we derive a
closed-form expression for the asymptotic sum-rate of the proposed strategy
in the high SNR regime.

\subsection{Asymptotic Sum-Rate as SNR $\rightarrow \infty $}

Here, we analyze the uplink achievable sum-rate%
\begin{equation}
R^{SD}=\dsum\limits_{m\in \left\{ A,B\right\} }R_{m}^{CD}+R_{m}^{PNC}
\end{equation}%
as the SNRs, i.e., $\frac{P_{A}}{N_{0}}$ and $\frac{P_{B}}{N_{0}}$, tend to
infinity. It is known that, in the high SNR regime, equal power allocation
is asymptotically optimal. Then, the upper bound of the uplink sum-rate of
the MIMO TWRC is given by (cf., (\ref{Outer Bound}))%
\begin{equation}
R^{UL}\approx \frac{1}{2}\dsum\limits_{m\in \left\{ A,B\right\} }\log
\left\vert \mathbf{I}_{n_{R}}+\frac{P_{m}}{N_{0}n_{m}}\mathbf{H}_{m}\mathbf{H%
}_{m}^{\dagger }\right\vert  \label{UpperBound}
\end{equation}%
where \textquotedblleft $x\approx y$\textquotedblright\ means%
\begin{equation}
\lim_{SNR\longrightarrow \infty }\left( x-y\right) =0. \notag
\end{equation}

Now, we present the following theorem on the asymptotic sum-rate of the
proposed scheme. Denote by $R^{SD}$ the uplink achievable sum-rate of the
proposed space-division scheme.

\begin{theo}
\label{Theorem 3}For a given $l%
{\acute{}}%
$, the uplink achievable sum-rate of the proposed space-division scheme
satisfies%
\begin{subequations}
\begin{equation}
\lim_{SNR\longrightarrow \infty }R^{UL}-R^{SD}=\Delta ^{SD}
\label{Inequality}
\end{equation}%
where%
\begin{equation}
\Delta ^{SD}\triangleq -\log \dprod\limits_{i=k+1}^{k+l%
{\acute{}}%
}\frac{\lambda _{i}}{2}-\log \dprod\limits_{i=k+l%
{\acute{}}%
+1}^{k+l}\sqrt{\lambda _{i}(2-\lambda _{i})}\geq 0.  \label{Delta_SD}
\end{equation}
\end{subequations}
\end{theo}

The proof of Theorem 3 can be found in Appendix B. Notice that the first
term in (\ref{Delta_SD}), i.e., $-\log \dprod\limits_{i=k+1}^{k+l%
{\acute{}}%
}\frac{\lambda _{i}}{2}$, is the rate loss incurred by
the PNC spatial streams, and the second term, i.e., $\log \dprod\limits_{i=k+l%
{\acute{}}%
+1}^{k+l}\sqrt{\lambda _{i}(2-\lambda _{i})}$, is that incurred by the
complete-decoding spatial streams.

\begin{rema}
For the case of $n_{A},n_{B}\geq n_{R}$, we have $l=0$ and $\lambda _{i}=2$
for $i=1,...,k$. (See Remark (\ref{Remark 4}).) Then, from (\ref{Delta_SD}),
we have $\Delta ^{SD}=0$, which means that the scheme is asymptotically
optimal. This agrees with the fact that our proposed space-division scheme
reduces to the GSVD scheme which is indeed asymptotically optimal in the
high SNR regime \cite{HJYangIT11}.
\end{rema}

\begin{coro}
The optimal $l%
{\acute{}}%
$ to minimize the rate gap $\Delta ^{SD}$ in (\ref{Delta_SD}) satisfies%
\begin{equation}
2>\lambda _{k+1}\geq ...\geq \lambda _{k+l%
{\acute{}}%
}\geq \frac{8}{5}>\lambda _{k+l%
{\acute{}}%
+1}\geq ...\geq \lambda _{k+l}>1.  \label{Range1}
\end{equation}%
With this choice of $l%
{\acute{}}%
$, the asymptotic rate gap $\Delta ^{SD}$ is at most $l\log (5/4)$ bits,
which occurs when $\lambda _{k+1}=\lambda _{k+2}=...=\lambda _{k+l}=\frac{8}{%
5}$.
\end{coro}

\begin{rema}
From the above corollary, the asymptotic gap to the sum-capacity upper bound
is $l\log (5/4)$ bits for the worst case. Noting $l\leq n_{m},m\in \left\{
A,B\right\} $, we see that the gap is at most $\min \{n_{A},n_{B}\}\log (5/4)$ bits, or $\frac{1}{2}\log (5/4)\approx
0.161$ bits per user-antenna.
\end{rema}

\subsection{Average Sum-Rate via Large-System Analysis}

In this subsection, we investigate the statistical average of the rate gap $%
\Delta ^{SD}$ in fading channels. To this end, the distribution of \{$\lambda _{i}$\}, i.e., the eigenvalues of $%
\mathbf{U}_{A}\mathbf{U}_{A}^{\dagger }+\mathbf{U}_{B}\mathbf{U}%
_{B}^{\dagger }$, is required. However, such a distribution is difficult to
obtain in general. Here, we employ the large-system analysis to find an
approximation of the distribution of \{$\lambda _{i}$\}. The distribution
obtained in this way becomes exact as the number of antennas in the
system is large.

We assume Rayleigh fading, in which the channel coefficients are i.i.d.
circularly symmetric complex Gaussian random variables. Then, the matrices $%
\mathbf{U}_{A}$ and $\mathbf{U}_{B}$ in (\ref{Channel_Decomp}) are truncated
uniformly distributed unitary matrices, or alternatively, are asymptotically
free random matrices \cite{TulinoTextBook}. Thus, we can use the theory of
free probabilities to derive the asymptotic eigenvalue distribution (a.e.d.)
of $\mathbf{U}_{A}\mathbf{U}_{A}^{\dagger }+\mathbf{U}_{B}\mathbf{U}%
_{B}^{\dagger }$ as $n_{R}$ tends to infinity, with the result given in the
lemma below. Define $\eta _{m}\triangleq \frac{n_{m}}{n_{R}},m\in \left\{
A,B\right\} $.

\begin{lemm}
\label{Lemma 4}As $n_{R}\rightarrow \infty $ with $\frac{n_{A}}{n_{R}}%
\rightarrow \eta _{A}$ and $\frac{n_{B}}{n_{R}}\rightarrow \eta _{B}$, the a.e.d. of $\mathbf{U}_{A}\mathbf{U}_{A}^{\dagger }+\mathbf{U}%
_{B}\mathbf{U}_{B}^{\dagger }$ is given by%
\begin{eqnarray}
\mathcal{F}\left( \lambda ;\eta _{A},\eta _{B}\right) &=&\left[ 1-\eta
_{A}-\eta _{B}\right] ^{+}\delta \left( \lambda \right) +\left\vert \eta
_{A}-\eta _{B}\right\vert \delta \left( \lambda -1\right) +\left[ \eta
_{A}+\eta _{B}-1\right] ^{+}\delta \left( \lambda -2\right)  \notag \\
&&+\frac{1}{\pi }\func{Im}\left[ \frac{\sqrt{(1-\eta _{A}-\eta
_{B})^{2}-\left( 2\lambda -\lambda ^{2}\right) \left( 1- \left( \frac{\eta
_{A}-\eta _{B}}{\lambda -1}\right) ^{2}\right) }}{2\lambda -\lambda ^{2}}%
\right]  \label{Distribution}
\end{eqnarray}%
where $\delta \left( \cdot \right) $ is a Dirac delta function and $\func{Im}%
\left[ \cdot\right] $ is the imaginary part of a complex number.
\end{lemm}

The proof of the above lemma can be found in Appendix C. As $%
n_{R}\rightarrow \infty $, we see that for $\eta _{A}+\eta _{B}\geq 1$, the
portion of eigenvalues $\left\{ \lambda _{i}\right\} $ equal to 2 is given
by $\eta _{A}+\eta _{B}-1$. This portion corresponds to the dimension of the common space $\mathcal{S}_{A\parallel B}$ of $\textbf{H}_{AR}$ and $\textbf{H}_{BR}$. In addition, for $\eta _{A}\neq \eta _{B}$, the
portion of eigenvalues $\left\{ \lambda _{i}\right\} $ equal to 1 is given
by $\left\vert \eta _{A}-\eta _{B}\right\vert $. This portion corresponds to the dimension of $\mathcal{S}_{A\perp B}$ if $\eta_A \geq \eta_B$ or the dimension of $\mathcal{S}_{B\perp A}$ if $\eta_A < \eta_B$.

We are now ready to present the following asymptotic result.

\begin{theo}
\label{Theorem 4}As $n_{R}\rightarrow \infty $ with $\frac{n_{A}}{n_{R}}%
\rightarrow \eta _{A}$ and $\frac{n_{B}}{n_{R}}\rightarrow \eta _{B}$, the gap to the capacity upper bound satisfies
\begin{equation}
r^{SD} \triangleq \lim_{n_{R}\rightarrow \infty }\text{ }\frac{\Delta ^{SD}}{n_{R}}\text{ }%
=-\left( \int_{1}^{\frac{8}{5}}\log \sqrt{\lambda (2-\lambda )}+\int_{\frac{8%
}{5}}^{2}\log \frac{\lambda }{2}\right) \mathcal{F}(\lambda ;\eta _{A},\eta
_{B})d\lambda .  \label{Limit1}
\end{equation}
\end{theo}

\begin{proof}
The a.e.d. of $\lambda _{i}$ is given by Lemma \ref{Lemma 4}. Then, letting $%
n_{R}$ tends to infinity in (\ref{Delta_SD}), we immediately obtain the
theorem.
\end{proof}

Let $\overline{R}^{UL}$ be the average sum-capacity upper bound. Then, for a
large $n_{R}$, the average sum-rate of the proposed SD scheme can be first-order
approximated as%
\begin{equation}\label{approx}
\overline{R}^{SD}=\overline{R}^{UL}-n_{R}r^{SD}
\end{equation}%
with $r^{SD}$ given in (\ref{Limit1}).

We next study the symmetric case that the two users are equipped with the
same number of antennas, i.e., $\eta _{A}=\eta _{B}=\eta $.

\begin{coro}
\label{Corollary 5} For $0\leq \eta \leq \frac{1}{10}$,%
\begin{subequations}
\begin{equation}
r^{SD} %
=-\int_{1}^{\lambda ^{\ast }(\eta)}\log \sqrt{\lambda (2-\lambda )}\mathcal{G}%
(\lambda ;\eta )d\lambda ;  \label{Limit2}
\end{equation}%
for $\frac{1}{10}<\eta \leq 1$,%
\begin{equation}
r^{SD}=-\left(
\int_{1}^{\frac{8}{5}}\log \sqrt{\lambda (2-\lambda )}+\int_{\frac{8}{5}%
}^{\lambda ^{\ast }(\eta )}\log \frac{\lambda }{2}\right) \mathcal{G}%
(\lambda ;\eta )d\lambda ,  \label{Limit3}
\end{equation}%
where $\lambda ^{\ast }(\eta )=1+\sqrt{1-\left( 1-2\eta \right) ^{2}}$ and%
\begin{equation}
\mathcal{G}\left( \lambda ;\eta \right) =\frac{1}{\pi }\frac{\sqrt{\left(
2\lambda -\lambda ^{2}\right) -(1-2\eta )^{2}}}{2\lambda -\lambda ^{2}}.
\label{G_lemda}
\end{equation}
\end{subequations}
\end{coro}

\begin{proof}
Letting $\eta _{A}=\eta _{B}=\eta $, we obtain that $\mathcal{F}\left(
\lambda ;\eta _{A},\eta _{B}\right) =\mathcal{G}\left( \lambda ;\eta \right)
$ for $1<\lambda <\lambda ^{\ast },$ and $\mathcal{F}\left( \lambda ;\eta
_{A},\eta _{B}\right) =0$ for $\lambda ^{\ast }<\lambda <2$. In addition, $%
\lambda ^{\ast }(\eta)=\frac{8}{5}$ implies $\eta =\frac{1}{10}$. Based on these
facts and Theorem \ref{Theorem 4}, we obtain the corollary.
\end{proof}

\begin{rema}
From the above, we see that, if $\eta \leq \frac{1}{10},$ the probability of
$\lambda _{i}>\frac{8}{5}$ approaches zero as $n_{R}\rightarrow \infty $,
implying that complete decoding achieves a higher rate than PNC
for all spatial streams.
\end{rema}

\begin{coro}
\label{Corollary 6}The asymptotic normalized rate gap $r^{SD}$ in (\ref%
{Limit1}) is maximized at $\eta _{A}=\eta _{B}=1/2$, with the maximum given
by%
\begin{equation}
-\frac{1}{\pi }\left( \int_{1}^{\frac{8}{5}}\frac{\log \sqrt{\lambda
(2-\lambda )}}{\sqrt{2\lambda -\lambda ^{2}}}d\lambda +\int_{\frac{8}{5}}^{2}%
\frac{\log \frac{\lambda }{2}}{\sqrt{2\lambda -\lambda ^{2}}}d\lambda
\right) \approx 0.053\text{ bit.}
\end{equation}
\end{coro}

\begin{proof}
We first consider optimizing $\eta _{A}$ and $\eta _{B}$ under the
constraint of $\eta _{A}+\eta _{B}=2\eta $. From (\ref{Distribution}), we
see that, for any $\lambda \in (1,2)$, $\mathcal{F}\left( \lambda ;\eta
_{A},\eta _{B}\right) $ is maximized at $\eta _{A}=\eta _{B}=\eta $, and so is $r^{SD}$.

What remains is to optimize $\eta $. From (\ref{G_lemda}), $\mathcal{G}\left( \lambda ;\eta \right) $ is
maximized at $\eta =1/2$. Therefore, $r^{SD}$ is
maximized at $\eta =1/2$, which completes the proof.
\end{proof}

Fig. \ref{SD_rate_gp} illustrates the function of the normalized asymptotic
rate gap $r^{SD}$ against $\eta $. From Fig. \ref{SD_rate_gp}, this rate gap is maximized at $%
\eta =1/2$, which verifies Corollary \ref{Corollary 6}. Also, this rate gap
vanishes as $\eta $ tends to 0, implying that, for any fixed $n_{A}=n_{B}$,
the proposed space-division scheme can achieve the asymptotic capacity as $%
n_{R}$ tends to infinity. Moreover, this rate gap vanishes as $\eta $ tends
to 1. This agrees with the fact that, for $\eta \geq 1$, or equivalently, $%
n_{A}=n_{B}\geq n_{R}$, the proposed space-division based scheme reduces to
the GSVD scheme in \cite{HJYangIT11}.

\section{Numerical Results}

In this section, we provide numerical results to evaluate the performance of
the proposed space-division based network-coding strategy for MIMO TWRCs.
The results presented below are obtained by averaging over 1,0000 random
channel realizations. Rayleigh-fading is assumed, i.e., the coefficients in
the channel matrices are independently and identically drawn from $\mathcal{N%
}_{c}(0,1)$.

We first present the numerical results for a MIMO TWRC of $n_{A}=n_{B}=2$
and $n_{R}=4$ in Fig. \ref{Figure nT2nR4}. The sum-capacity upper bound
(UB), the proposed space-division (SD) scheme, the GSVD scheme in \cite%
{HJYangIT11} and the complete-decoding scheme in \cite{YangIT11} are
included for comparison. We see that, at a relatively high SNR, e.g., SNR =
25 dB, the rate gap between the proposed SD scheme and the sum-capacity
upper bound is about 0.15 bit/channel-use, which is almost unnoticeable. We
also plot the high-SNR analytical result in (\ref{approx}) of the proposed SD scheme.
We observe that our analytical result are very tight for SNRs greater than
10 dB. From this figure, it is clear that the proposed SD scheme\
significantly outperforms the other schemes in the entire SNR range of
interest. For example, at the rate of $14$ bits per channel use, the
proposed SD scheme outperforms the complete-decoding and GSVD schemes by
more than $2.4$ dB. The slope of the achievable sum-rate curve is parallel
to that of the capacity upper bound, which implies that the proposed SD
scheme achieves full multiplexing gain.

In Fig. \ref{Figure nT2nR3}, we present the numerical results for a MIMO
TWRC of $n_{A}=n_{B}=2$ and $n_{R}=3$. The same set of rate curves from
simulation as in Fig. \ref{Figure nT2nR4} are included for comparison.
Again, we see that the gap between the sum-rate of the proposed SD scheme
and the sum-capacity upper bound is almost unnoticeable at a relatively high
SNR, e.g., greater than 15 dB. The proposed SD scheme outperforms its
counterparts throughout the SNR range of interest.

In Figures \ref{Figure Eda1_2} and \ref{Figure Eda2_3}, we show the scaling
effect of the antennas on the average achievable sum-rates. We see that the
asymptotic rate gap between the proposed SD scheme and the sum-capacity
upper bound increases linearly as the increase of $n_{R}$ for fixed $\eta _{A}$
and $\eta _{B}$. For example, for the case of $\eta _{A}=\eta
_{B}=1/2$ in Fig. \ref{Figure Eda1_2}, the rate gap at SNR = 25 dB is 0.14
bits per channel use for $n_{R}=4$; 0.29 bits per channel use for $n_{R}=6$;
and 0.40 bits per channel use for $n_{R}=8$. These numerical results agree
well with the asymptotic results in Corollaries \ref{Corollary 5} and \ref%
{Corollary 6}.

In Fig. \ref{Figure CapacityRegion}, we show the achievable rate-region of
the proposed SD scheme. The capacity-region outer bound and the achievable
rate-region of the complete-decoding scheme are also included for
comparison. From Fig. \ref{Figure CapacityRegion}, the difference between
the achievable rate-region of the proposed SD scheme and the capacity region
outer bound is negligible for a relatively high SNR. We also see that the
proposed SD scheme can achieve rate-pairs that cannot be achieved by the
complete-decoding scheme.

\section{Conclusion}

In this paper, we developed a new joint channel decomposition for MIMO
TWRCs. Based on that, we proposed a space-division based network-coding
scheme with the achievable sum-rate within $\frac{1}{2}\log (5/4)\approx
0.161$ bit per user-antenna of the capacity upper bound in the high SNR
regime. We also show that, for Rayleigh-fading MIMO TWRCs, the average gap
between the achievable rate of the proposed scheme and the capacity upper
bound is no more than $0.053$ bit per relay-antenna in the high SNR regime. We
remark that this marginal gap is due to the fact that the complete-decoding
and PNC strategies, collectively, fail to achieve the asymptotic capacity
upper bound, even for the case of single-antenna users. To completely remove
this gap, more advanced multi-dimension PNC relaying strategies may be
required. Moreover, in this paper, channel state information is assumed to
be globally known by both the transmitter and receiver sides. It is of
theoretical, and more practical, interests to investigate how to efficiently
communicate over MIMO TWRCs where only the receiver-side channel state
information is available. We will look into these problems in our future
research.

\appendices
\section{Proof of Theorem \protect\ref{Theorem 1}}
Here we prove Theorem \ref{Theorem 1}. Let $\lambda _{i}$ be an eigenvalue
of $\mathbf{U}_{A}\mathbf{U}_{A}^{\dagger }+\mathbf{U}_{B}\mathbf{U}%
_{B}^{\dagger }$ and $\mathbf{u}_{i}$ be the corresponding unit-length eigenvector
satisfying%
\begin{equation}
\left( \mathbf{U}_{A}\mathbf{U}_{A}^{\dagger }+\mathbf{U}_{B}\mathbf{U}%
_{B}^{\dagger }\right) \mathbf{u}_{i}=\lambda _{i}\mathbf{u}_{i}.
\label{eigen}
\end{equation}%
We are interested in four cases of $\lambda _{i}$: (a) $\lambda _{i}=2$; (b)
$1<\lambda _{i}<2;$ (c) $\lambda _{i}=1$; and (d) $0<\lambda _{i}<1$.

For case (a), $\lambda _{i}=2$ implies that%
\begin{equation*}
\mathbf{U}_{A}\mathbf{U}_{A}^{\dagger }\mathbf{u}_{i}=\mathbf{u}_{i}\text{
and }\mathbf{U}_{B}\mathbf{U}_{B}^{\dagger }\mathbf{u}_{i}=\mathbf{u}_{i}%
\text{.}
\end{equation*}%
Thus, $\mathbf{u}_{i}$ lies in the common space of $\mathcal{C}(\mathbf{U}%
_{A})$ and $\mathcal{C}(\mathbf{U}_{B})$.

For case (c), we have%
\begin{subequations}
\begin{equation}
\mathbf{U}_{A}\mathbf{U}_{A}^{\dagger }\mathbf{u}_{i}=\mathbf{u}_{i}\text{
and }\mathbf{U}_{B}\mathbf{U}_{B}^{\dagger }\mathbf{u}_{i}=\mathbf{0}
\label{case c1}
\end{equation}%
or%
\begin{equation}
\mathbf{U}_{A}\mathbf{U}_{A}^{\dagger }\mathbf{u}_{i}=\mathbf{0}\text{ and }%
\mathbf{U}_{B}\mathbf{U}_{B}^{\dagger }\mathbf{u}_{i}=\mathbf{u}_{i}\text{.}
\label{case c2}
\end{equation}
\end{subequations}

We next show that the eigenvalues in case (b) and case (d) appear in a
pair-wise manner. Denote%
\begin{equation}
\mathbf{l}_{m;i}=\mathbf{U}_{m}\left( \mathbf{U}_{m}^{\dagger }\mathbf{U}%
_{m}\right) ^{-1}\mathbf{U}_{m}^{\dagger }\mathbf{u}_{i}=\mathbf{U}_{m}%
\mathbf{U}_{m}^{\dagger }\mathbf{u}_{i}.  \label{Eq 3}
\end{equation}%
Note that $\mathbf{l}_{m;i}$ is the projection of vector $\mathbf{u}_{i}$
onto the column space of $\mathbf{U}_{m}$. From (\ref{eigen}), we obtain%
\begin{equation}
\mathbf{u}_{i}=\frac{1}{\lambda _{i}}\left( \mathbf{l}_{A;i}+\mathbf{l}%
_{B;i}\right) .  \label{Eq 4}
\end{equation}%
The above implies that $\mathbf{u}_{i}$, $\mathbf{l}_{A;i}$ and $\mathbf{l}%
_{B;i}$ lie on the same two-dimension plane (denoted by $\mathcal{S}_{i}$).
We have the following facts.

\begin{lemm}
For any $\lambda _{i}$ in case (b), the corresponding $\mathbf{u}_{i}$ is
the angular bisector of $\mathbf{l}_{A;i}$ and $\mathbf{l}_{B;i}$, i.e.%
\begin{equation}
\left\Vert \mathbf{l}_{A;i}\right\Vert ^{2}=\mathbf{u}_{i}^{\dagger }\mathbf{%
l}_{A;i}=\mathbf{u}_{i}^{\dagger }\mathbf{l}_{B;i}=\left\Vert \mathbf{l}%
_{B;i}\right\Vert ^{2}.  \label{Eq 5}
\end{equation}
\end{lemm}

\begin{proof}
To prove the lemma, we first multiply both sides of (\ref{eigen}) by $%
\mathbf{U}_{A}\mathbf{U}_{A}^{\dagger }$. Then, after some straightforward
manipulations, we obtain%
\begin{equation}
\left( \lambda _{i}-1\right) \mathbf{U}_{A}\mathbf{U}_{A}^{\dagger }\mathbf{u%
}_{i}=\mathbf{U}_{A}\mathbf{U}_{A}^{\dagger }\mathbf{U}_{B}\mathbf{U}%
_{B}^{\dagger }\mathbf{u}_{i}.  \label{Eq 6}
\end{equation}%
Similarly, we have
\begin{equation}
\left( \lambda _{i}-1\right) \mathbf{U}_{B}\mathbf{U}_{B}^{\dagger }\mathbf{u%
}_{i}=\mathbf{U}_{B}\mathbf{U}_{B}^{\dagger }\mathbf{U}_{A}\mathbf{U}%
_{A}^{\dagger }\mathbf{u}_{i}.  \label{Eq 7}
\end{equation}%
Then,
\begin{align*}
& \mathbf{u}_{i}^{\dagger }\mathbf{l}_{A;i}\text{ }\overset{(a)}{=}\mathbf{u}%
_{i}^{\dagger }\mathbf{U}_{A}\mathbf{U}_{A}^{\dagger }\mathbf{u}_{i}\overset{%
(b)}{=}\frac{1}{\lambda _{i}-1}\mathbf{u}_{i}^{\dagger }\mathbf{U}_{A}%
\mathbf{U}_{A}^{\dagger }\mathbf{U}_{B}\mathbf{U}_{B}^{\dagger }\mathbf{u}%
_{i} \\
& \overset{(c)}{=}\frac{1}{\lambda _{i}-1}\mathbf{u}_{i}^{\dagger }\mathbf{U}%
_{B}\mathbf{U}_{B}^{\dagger }\mathbf{U}_{A}\mathbf{U}_{A}^{\dagger }\mathbf{u%
}_{i}\overset{(d)}{=}\mathbf{u}_{i}^{\dagger }\mathbf{U}_{B}\mathbf{U}%
_{B}^{\dagger }\mathbf{u}_{i}\overset{(e)}{=}\mathbf{u}_{i}^{\dagger }%
\mathbf{l}_{B;i}\text{ }
\end{align*}%
where step $(a)$ follows from (\ref{Eq 3}), $(b)$ from (\ref{Eq 6}), $(c)$
from the fact that the Hermitian transpose of a real-valued scalar is itself, $(d)$ from (\ref{Eq
7}), and $(e)$ again from (\ref{Eq 3}). From (\ref{Eq 3}), the projection of
$\mathbf{u}_{i}$ onto $\mathbf{l}_{m;i}$ is just $\mathbf{l}_{m;i}$. Thus, $%
\left\Vert \mathbf{l}_{m;i}\right\Vert ^{2}=\mathbf{u}_{i}^{\dagger }\mathbf{%
l}_{m;i}$, which completes the proof.
\end{proof}

\begin{lemm}
For any $\lambda _{i}$ $\in (1,2)$ (as in case (b)), $\lambda _{i}^{\prime
}=2-\lambda _{i}$ is also an eigenvalue of $\mathbf{U}_{A}\mathbf{U}%
_{A}^{\dagger }+\mathbf{U}_{B}\mathbf{U}_{B}^{\dagger }$, and the
corresponding unit-length eigenvector is given by%
\begin{equation}
\mathbf{u}_{i}^{\prime }=\frac{1}{\sqrt{\lambda _{i}\lambda _{i}^{\prime }}}%
\left( \mathbf{l}_{A;i}-\mathbf{l}_{B;i}\right) .  \label{Eq 9}
\end{equation}
\end{lemm}

\begin{proof}
By definition, we have%
\begin{align}
& \left( \mathbf{U}_{A}\mathbf{U}_{A}^{\dagger }+\mathbf{U}_{B}\mathbf{U}%
_{B}^{\dagger }\right) \mathbf{u}_{i}^{\prime }  \notag \\
& \overset{(a)}{=}\frac{1}{\sqrt{\lambda _{i}\lambda _{i}^{\prime }}}\left(
\mathbf{U}_{A}\mathbf{U}_{A}^{\dagger }+\mathbf{U}_{B}\mathbf{U}%
_{B}^{\dagger }\right) \left( \mathbf{l}_{A;i}-\mathbf{l}_{B;i}\right)
\notag \\
& \overset{(b)}{=}\frac{1}{\sqrt{\lambda _{i}\lambda _{i}^{\prime }}}\left(
\mathbf{U}_{A}\mathbf{U}_{A}^{\dagger }\mathbf{u}_{i}+\left( \lambda
_{i}-1\right) \mathbf{U}_{B}\mathbf{U}_{B}^{\dagger }\mathbf{u}_{i}-\left(
\lambda _{i}-1\right) \mathbf{U}_{A}\mathbf{U}_{A}^{\dagger }\mathbf{u}_{i}-%
\mathbf{U}_{B}\mathbf{U}_{B}^{\dagger }\mathbf{u}_{i}\right)  \notag \\
& =\frac{\lambda _{i}^{\prime }}{\sqrt{\lambda _{i}\lambda _{i}^{\prime }}}%
\left( \mathbf{l}_{A;i}-\mathbf{l}_{B;i}\right) =\lambda _{i}^{\prime }%
\mathbf{u}_{i}^{\prime }
\end{align}%
where step $(a)$ follows from (\ref{Eq 9}), and step $(b)$ from (\ref{Eq 3}%
), (\ref{Eq 6}) and (\ref{Eq 7}).

What remains is to show that $\left\Vert \mathbf{u}_{i}^{\prime }\right\Vert
=1$. To see this, we left-multiply both sides of (\ref{eigen}) by $\mathbf{u}%
_{i}^{\dag }$, yielding%
\begin{equation}
\left\Vert \mathbf{l}_{A;i}\right\Vert ^{2}+\left\Vert \mathbf{l}%
_{B;i}\right\Vert ^{2}=\lambda _{i}.  \label{Eq 8}
\end{equation}%
Together with (\ref{Eq 5}), we obtain%
\begin{equation}
\left\Vert \mathbf{l}_{A;i}\right\Vert ^{2}=\left\Vert \mathbf{l}%
_{B;i}\right\Vert ^{2}=\frac{\lambda _{i}}{2}.  \label{Eq 2}
\end{equation}%
Moreover, left multiplying (\ref{Eq 6}) and (\ref{Eq 7}) respectively by $%
\mathbf{u}_{i}^{\dagger }$ and plugging in (\ref{Eq 3}), we obtain%
\begin{equation}
\mathbf{l}_{A;i}^{\dag }\mathbf{l}_{B;i}=(\lambda _{i}-1)\left\Vert \mathbf{l%
}_{A;i}\right\Vert ^{2}=(\lambda _{i}-1)\left\Vert \mathbf{l}%
_{B;i}\right\Vert ^{2}=\mathbf{l}_{B;i}^{\dag }\mathbf{l}_{A;i}.
\label{Eq 10}
\end{equation}%
Then%
\begin{align}
\mathbf{u}_{i}^{\prime \dag }\mathbf{u}_{i}^{\prime }& =\frac{1}{\lambda
_{i}\lambda _{i}^{\prime }}\left( \mathbf{l}_{A;i}-\mathbf{l}_{B;i}\right)
^{\dag }\left( \mathbf{l}_{A;i}-\mathbf{l}_{B;i}\right)  \notag \\
& =\frac{1}{\lambda _{i}\lambda _{i}^{\prime }}\left( \left\Vert \mathbf{l}%
_{A;i}\right\Vert ^{2}-\mathbf{l}_{A;i}^{\dag }\mathbf{l}_{B;i}-\mathbf{l}%
_{B;i}^{\dag }\mathbf{l}_{A;i}+\left\Vert \mathbf{l}_{B;i}\right\Vert
^{2}\right)  \notag \\
& \overset{(a)}{=}\frac{1}{\lambda _{i}\lambda _{i}^{\prime }}\left( \lambda
_{i}-(\lambda _{i}-1)\left\Vert \mathbf{l}_{A;i}\right\Vert ^{2}-(\lambda
_{i}-1)\left\Vert \mathbf{l}_{B;i}\right\Vert ^{2}\right) \notag\\
& \overset{(b)}{=}1
\end{align}%
where step ($a$) follows from (\ref{Eq 2}) and (\ref{Eq 10}), and step ($b$)
from (\ref{Eq 2}) and the fact of $\lambda _{i}^{\prime }=2-\lambda _{i}$.
This completes the proof.
\end{proof}

\begin{lemm}
The subspace $\mathcal{S}_{i}$ spanned by $\textbf{l}_{A;i}$ and $\textbf{l}_{B;i}$ is orthogonal to $\mathcal{S}%
_{j}$, for any $j$ $\neq i$.
\end{lemm}

\begin{proof}
From (\ref{Eq 4}) and (\ref{Eq 9}), we see that both $\mathbf{u}_{i}$ and $%
\mathbf{u}_{i}^{\prime }$ lie on the plane $\mathcal{S}_{i}$. As $\mathbf{u}%
_{i}$ and $\mathbf{u}_{i}^{\prime }$ are orthogonal to each other, $\mathcal{%
S}_{i}$ is spanned by $\mathbf{u}_{i}$ and $\mathbf{u}_{i}^{\prime }$. Then,
the lemma holds straightforwardly by noting the orthogonality between the
eigenvectors.
\end{proof}

Now we give an overall picture of the eigenvalues and eigenvectors of $%
\mathbf{U}_{A}\mathbf{U}_{A}^{\dagger }+\mathbf{U}_{B}\mathbf{U}%
_{B}^{\dagger }$. Denote the $k$ eigenvalues in case (a) by $\lambda
_{1},\cdots ,\lambda _{k}$, and the corresponding orthogonal eigenvectors by
$\mathbf{u}_{1},\cdots ,$ $\mathbf{u}_{k}$. Also denote the $l$ eigenvalues
in case (b) by $\lambda _{k+1},\cdots ,\lambda _{k+l}$ in the descending
order, and the corresponding eigenvectors by $\mathbf{u}_{k+1},\cdots ,$ $%
\mathbf{u}_{k+l}$. As the eigenvalues in (b) and (d) appears in a pair-wise
manner, we further denote the $l$ eigenvalues in case (d) by $\lambda
_{k+1}^{\prime },\cdots ,$ $\lambda _{k+l}^{\prime }$ in the descending
order, and the corresponding eigenvectors by $\mathbf{u}_{k+1}^{\prime
},\cdots ,$ $\mathbf{u}_{k+l}^{\prime }$. Moreover, we denote the $d_{A}$
orthogonal eigenvectors in case (c.1) by $\mathbf{u}_{k+l+1},\cdots ,\mathbf{%
u}_{k+l+d_{A}}$, and the $d_{B}$ orthogonal eigenvectors in case (c.2) by $%
\mathbf{u}_{k+l+d_{A}+1},\cdots ,\mathbf{u}_{k+l+d_{A}+d_{B}}$. Let%
\begin{equation}
\mathbf{U}=\left[ \mathbf{u}_{1},\cdots ,\mathbf{u}_{k},\mathbf{u}_{k+1},%
\mathbf{u}_{k+1}^{\prime },\cdots ,\mathbf{u}_{k+l},\mathbf{u}_{k+l}^{\prime
},\mathbf{u}_{k+l+1},...,\mathbf{u}_{k+l+d_{A}+d_{B}}\right] .  \label{Def_U}
\end{equation}%
It can be readily verified that $\mathbf{U}$ is an orthonormal matrix
satisfying $\mathbf{U}_{{}}^{\dag }\mathbf{U}=\mathbf{I}_{n_{A}+n_{B}-k}$.
Define%
\begin{subequations}
\begin{equation}
\mathbf{U}_{A}^{\prime }=\left[ \mathbf{u}_{1},\cdots ,\mathbf{u}_{k},\frac{%
\mathbf{l}_{A;k+1}}{\left\Vert \mathbf{l}_{A;k+1}\right\Vert },\cdots ,\frac{%
\mathbf{l}_{A;k+l}}{\left\Vert \mathbf{l}_{A;k+l}\right\Vert },\mathbf{u}%
_{k+l+1},...,\mathbf{u}_{k+l+d_{A}}\right]
\end{equation}%
and%
\begin{equation}
\mathbf{U}_{B}^{\prime }=\left[ \mathbf{u}_{1},\cdots ,\mathbf{u}_{k},\frac{%
\mathbf{l}_{B;k+1}}{\left\Vert \mathbf{l}_{B;k+1}\right\Vert },\cdots ,\frac{%
\mathbf{l}_{B;k+l}}{\left\Vert \mathbf{l}_{B;k+l}\right\Vert },\mathbf{u}%
_{k+l+d_{A}+1},...,\mathbf{u}_{k+l+d_{A}+d_{B}}\right] .
\end{equation}%
\end{subequations}
In the above, $\mathbf{u}_{1},\cdots ,\mathbf{u}_{k}$ are the eigenvectors
in case (a); $\mathbf{u}_{k+l+1},...,\mathbf{u}_{k+l+d_{A}}$ are the
eigenvectors in case (c) satisfying $\mathbf{U}_{A}\mathbf{U}_{A}^{\dag }%
\mathbf{u}_{i}=\mathbf{u}_{i},$ for $i=k+l+1,...,k+l+d_{A}$; $\mathbf{u}%
_{k+l+d_{A}+1},\cdots ,$ $\mathbf{u}_{k+l+d_{A}+d_{B}}$ are the eigenvectors
in case (c) satisfying $\mathbf{U}_{B}\mathbf{U}_{B}^{\dag }\mathbf{u}_{i}=%
\mathbf{u}_{i},$ for $i=k+l+d_{A}+1,...,k+l+d_{A}+d_{B}$.

Then, from Lemmas 3 and 4, it can be verified that $\mathbf{D}_{m}$ in (\ref%
{Dm}) satisfies%
\begin{equation}
\mathbf{U}_{m}^{\prime }=\mathbf{UD}_{m},m\in \left\{ A,B\right\} .
\label{Dm2}
\end{equation}

From Lemma 3 and the fact that $\mathbf{l}_{m;i}\in \mathcal{S}_{i}$ for $%
i=k+1,\cdots ,k+l$, the columns of $\mathbf{U}_{m}^{\prime }$ are
orthonormal. Together with the fact that all columns of $\mathbf{U}%
_{m}^{\prime }$ lie in the columnspace of $\mathbf{U}_{m}$ (and so in the
columnspace of $\mathbf{H}_{m}$), we see that $\mathbf{H}_{m}$ and $\mathbf{%
U}_{m}^{\prime }$ share the same columnspace. Thus, there exists an $n_{m}$%
-by-$n_{m}$ square matrix $\mathbf{G}_{m}$ such that%
\begin{equation}
\mathbf{H}_{mR}=\mathbf{U}_{m}^{\prime }\mathbf{G}_{m}.  \label{Eq 13}
\end{equation}%
Combining (\ref{Dm2}) and (\ref{Eq 13}), we obtain%
\begin{equation}
\mathbf{H}_{mR}=\mathbf{UD}_{m}\mathbf{G}_{m}  \label{Eq 14}
\end{equation}%
which completes the proof of Theorem \ref{Theorem 1}.

\section{Proof of Theorem \protect\ref{Theorem 3}}

We first consider the sum-rate upper bound:
\begin{eqnarray}
&&R^{UL}\overset{(a)}{\approx }\frac{1}{2}\dsum\limits_{m\in \left\{
A,B\right\} }\log \left\vert \mathbf{I}_{n_{R}}+\frac{P_{m}}{N_{0}n_{m}}%
\mathbf{UD}_{m}\mathbf{R}_{m}\mathbf{R}_{m}^{\dag }\mathbf{D}_{m}^{\dagger }%
\mathbf{U}^{\dagger }\right\vert  \notag \\
&&\overset{(b)}{=}\frac{1}{2}\dsum\limits_{m\in \left\{ A,B\right\} }\log
\left\vert \mathbf{I}_{n_{m}}+\frac{P_{m}}{N_{0}n_{m}}\mathbf{R}_{m}\mathbf{R%
}_{m}^{\dag }\right\vert  \notag \\
&&\overset{(c)}{\approx }\frac{1}{2}\dsum\limits_{m\in \left\{ A,B\right\}
}\log \left\vert \frac{P_{m}}{N_{0}n_{m}}\mathbf{R}_{m}\mathbf{R}_{m}^{\dag
}\right\vert  \label{Upp}
\end{eqnarray}%
where step (\textit{a}) follows by substituting (\ref{Decomp}) into (\ref%
{UpperBound}), step (\textit{b}) follows from the facts that $\mathbf{D}%
_{m}^{\dagger }\mathbf{U}^{\dagger }\mathbf{U\mathbf{D}}_{m}\mathbf{=I}%
_{n_{m}}$ and $\left\vert \mathbf{I}+\mathbf{AB}\right\vert =\left\vert
\mathbf{I}+\mathbf{BA}\right\vert $, and step (\textit{c}) utilizes the fact
that $\mathbf{R}_{m}$ is a square matrix.

Now we consider the achievable sum-rate of the proposed space-division
scheme. For notational simplicity, let $\mathbf{H}_{m}^{CD}=\mathbf{D}%
_{m;2,2}\mathbf{R}_{m;2,2}$, $m\in \left\{ A,B\right\} $. From (\ref{Rate CD}%
), the sum-rate of the complete-decoding spatial streams can be expressed as%
\begin{eqnarray}
R_{A}^{CD}+R_{B}^{CD}\overset{(a)}{=}\frac{1}{2}\log \left\vert \mathbf{I}%
+\dsum\limits_{m\in \left\{ A,B\right\} }\frac{P_{m}}{N_{0}n_{m}}\mathbf{H}%
_{m}^{CD}(\mathbf{H}_{m}^{CD})^{\dag }\right\vert \notag
\end{eqnarray}
\begin{eqnarray*}
&&\overset{(b)}{=}\frac{1}{2}\log \left\vert \mathbf{I}+\frac{P_{A}}{%
N_{0}n_{A}}\mathbf{H}_{A}^{CD}(\mathbf{H}_{A}^{CD})^{\dag }\right\vert +%
\frac{1}{2}\log \frac{\left\vert \mathbf{I}+\dsum\limits_{m\in \left\{
A,B\right\} }\frac{P_{m}}{N_{0}n_{m}}\mathbf{H}_{m}^{CD}(\mathbf{H}%
_{m}^{CD})^{\dag }\right\vert }{\left\vert \mathbf{I}+\frac{P_{A}}{N_{0}n_{A}%
}\mathbf{H}_{A}^{CD}(\mathbf{H}_{A}^{CD})^{\dag }\right\vert } \\
&&\overset{(c)}{\approx }\frac{1}{2}\log \left\vert \frac{P_{A}}{N_{0}n_{A}}%
\mathbf{R}_{A;2,2}\mathbf{R}_{A;2,2}^{\dag }\right\vert +\frac{1}{2}\log
\left\vert \frac{P_{B}}{N_{0}n_{B}}(\mathbf{H}_{B}^{CD})^{\dag }\left(
\mathbf{I}+\frac{P_{A}}{N_{0}n_{A}}\mathbf{H}_{A}^{CD}(\mathbf{H}%
_{A}^{CD})^{\dag }\right) ^{-1}\mathbf{H}_{B}^{CD}\right\vert \\
&&\overset{(d)}{=}\dsum\limits_{m\in \left\{ A,B\right\} }\frac{1}{2}\log
\left\vert \frac{P_{m}}{N_{0}n_{m}}\mathbf{R}_{m;2,2}\mathbf{R}%
_{m;2,2}^{\dag }\right\vert +\frac{1}{2}\log \left\vert \mathbf{D}%
_{B;2,2}^{\dag }\left( \mathbf{I}+\frac{P_{A}}{N_{0}n_{A}}\mathbf{H}%
_{A}^{CD}(\mathbf{H}_{A}^{CD})^{\dag }\right) ^{-1}\mathbf{D}%
_{B;2,2}\right\vert
\end{eqnarray*}%
where step (\textit{a}) utilizes the fact that equal power allocation is
asymptotically optimal, and step (d) follows by substituting $\mathbf{H}%
_{B}^{CD}=\mathbf{D}_{B;2,2}\mathbf{R}_{B;2,2}$. Applying the matrix
inversion lemma to $\left( \mathbf{I}+\frac{P_{A}}{N_{0}n_{A}}\mathbf{H}%
_{A}^{CD}(\mathbf{H}_{A}^{CD})^{\dag }\right) ^{-1}$, we further obtain%
\begin{eqnarray}
&&R_{A}^{CD}+R_{B}^{CD}=\dsum\limits_{m\in \left\{ A,B\right\} }\frac{1}{2}%
\log \left\vert \frac{P_{m}}{N_{0}n_{m}}\mathbf{R}_{m;2,2}\mathbf{R}%
_{m;2,2}^{\dag }\right\vert  \notag \\
&&\text{ \ \ \ \ \ }+\frac{1}{2}\log \left\vert \mathbf{I}-\frac{P_{A}}{%
N_{0}n_{A}}\mathbf{D}_{B;2,2}^{\dag }\mathbf{H}_{A}^{CD}\left( \mathbf{I}+%
\frac{P_{A}}{N_{0}n_{A}}(\mathbf{H}_{A}^{CD})^{\dag }\mathbf{H}%
_{A}^{CD}\right) ^{-1}(\mathbf{H}_{A}^{CD})^{\dag }\mathbf{D}%
_{B;2,2}\right\vert  \notag \\
&&\overset{(a)}{\approx }\dsum\limits_{m\in \left\{ A,B\right\} }\frac{1}{2}%
\log \left\vert \frac{P_{m}}{N_{0}n_{m}}\mathbf{R}_{m;2,2}\mathbf{R}%
_{m;2,2}^{\dag }\right\vert +\frac{1}{2}\log \left\vert \mathbf{I}-\mathbf{D}%
_{B;2,2}^{\dag }\mathbf{D}_{A;2,2}\mathbf{D}_{A;2,2}^{\dag }\mathbf{D}%
_{B;2,2}\right\vert  \notag \\
&&\overset{(b)}{=}\dsum\limits_{m\in \left\{ A,B\right\} }\frac{1}{2}\log
\left\vert \frac{P_{m}}{N_{0}n_{m}}\mathbf{R}_{m;2,2}\mathbf{R}%
_{m;2,2}^{\dag }\right\vert +\frac{1}{2}\log \dprod\limits_{i=k+l%
{\acute{}}%
+1}^{k+l}\lambda _{i}(2-\lambda _{i})  \label{Con_Coded}
\end{eqnarray}%
where step (\textit{a}) follows by noting $\mathbf{I}+\frac{P_{A}}{N_{0}n_{A}%
}(\mathbf{H}_{A}^{CD})^{\dag }\mathbf{H}_{A}^{CD}\approx \frac{P_{A}}{%
N_{0}n_{A}}(\mathbf{H}_{A}^{CD})^{\dag }\mathbf{H}_{A}^{CD}$ and $\mathbf{H}_{A}^{CD}=\mathbf{D}%
_{A;2,2}\mathbf{R}_{A;2,2}$, $m\in \left\{ A,B\right\} $, and step (%
\textit{b}) utilizes the definitions in (\ref{Dm}) and (\ref{Dm3}).
Moreover, letting $w_{A}=w_{B}=1$ in (\ref{WSR1}), we obtain\ the sum-rate
of the PNC spatial streams as%
\begin{equation}
R_{A}^{PNC}+R_{B}^{PNC}=\frac{1}{2}\dsum\limits_{m\in \left\{ A,B\right\}
}\left(\log \left\vert \frac{P_{m}}{N_{0}n_{m}}\mathbf{R}_{m;1,1}\mathbf{R}%
_{m;1,1}^{\dag }\right\vert +\log \left\vert \mathbf{\tilde{D}}_{m;1,1}%
\mathbf{\tilde{D}}_{m;1,1}^{\dag }\right\vert\right) \text{.}  \label{PNC_Coded}
\end{equation}%
From (\ref{dd}), $\mathbf{p}_{i}$ is the angular bisection of $\mathbf{e}%
_{A;i}$ and $\mathbf{e}_{B;i}$, or equivalently, $\mathbf{p}_{i}=[1,0]^{T}$,
for the sum-rate case of $w_{A}=w_{B}=1$. Then, using the definition in (\ref%
{ddd}), we obtain%
\begin{equation}
\log \left\vert \mathbf{\tilde{D}}_{m;1,1}\mathbf{\tilde{D}}_{m;1,1}^{\dag
}\right\vert =\dsum\limits_{i=k+1}^{k+l%
{\acute{}}%
}\log \frac{\lambda _{i}}{2}.  \label{DD}
\end{equation}%
Combining (\ref{Upp})-(\ref{DD}), we complete the
proof of Theorem 3.

\section{Proof of Lemma \protect\ref{Lemma 4}}

We prove by using the theory of free probability \cite{Voiculescu93}. The
a.e.d. of $\mathbf{U}_{m}\mathbf{U}_{m}^{\dagger }$ is given by
\begin{equation*}
p_{m}\left( \lambda \right) =\eta _{m}\delta \left( \lambda -1\right)
+\left( 1-\eta _{m}\right) \delta \left( \lambda \right) \text{, }m\in
\left\{ A,B\right\} .
\end{equation*}%
Let $X_{m}$ be a random variable with PDF $p_{m}\left( \lambda \right) $.
Its Stieltjes transform is given by (cf., (2.40) in \cite{TulinoTextBook})%
\begin{equation*}
S_{X_{m}}\left( z\right) =E\left[ \frac{1}{X_{m}-z}\right] =\frac{\eta _{m}}{%
1-z}-\frac{1-\eta _{m}}{z}.
\end{equation*}%
Then, the inverse function of $S_{X_{m}}\left( z\right) $ is given by%
\begin{equation*}
S_{X_{m}}^{-1}\left( s\right) =\frac{-\left( 1-s\right) \pm \sqrt{\left(
1-s\right) ^{2}-4s\left( \eta _{m}-1\right) }}{2s}.
\end{equation*}%
Using the relation between Stieltjes transform and R-transform (cf., (2.72)
in \cite{TulinoTextBook}), we obtain the R-transform of $X_{m}$ as
\begin{equation*}
R_{X_{m}}\left( z\right) =S_{X_{m}}^{-1}\left( -z\right) -\frac{1}{z}=\frac{%
z-1\mp \sqrt{\left( z-1\right) ^{2}+4\eta _{m}z}}{2z}.
\end{equation*}%
From Theorem 2.64 of \cite{TulinoTextBook}, as $\mathbf{U}_{A}\mathbf{U}%
_{A}^{\dagger }$ and $\mathbf{U}_{B}\mathbf{U}_{B}^{\dagger }$ are
asymptotically free random matrices, the R-transform
of the a.e.d. of $\mathbf{U}_{A}\mathbf{U}_{A}^{\dagger }+\mathbf{U}_{B}%
\mathbf{U}_{B}^{\dagger }$ is given by%
\begin{equation*}
R_{AB}\left( z\right) =R_{X_{A}}\left( z\right) +R_{X_{B}}\left( z\right)
=\dsum\limits_{m\in \left\{ A,B\right\} }\frac{z-1\mp \sqrt{\left(
z-1\right) ^{2}+4\eta _{m}z}}{2z}.
\end{equation*}%
Then, the Stieltjes transform of the a.e.d. of $\mathbf{U}_{A}\mathbf{U}%
_{A}^{\dagger }+\mathbf{U}_{B}\mathbf{U}_{B}^{\dagger }$ satisfies%
\begin{equation*}
S_{AB}^{-1}\left( -z\right) =1\mp \dsum\limits_{m\in \left\{ A,B\right\} }%
\frac{\sqrt{\left( z-1\right) ^{2}+4\eta _{m}z}}{2z}.
\end{equation*}%
Letting $y=S_{AB}^{-1}\left( -z\right) $ , we obtain
\begin{equation*}
\dsum\limits_{m\in \left\{ A,B\right\} }\sqrt{\left( z-1\right) ^{2}+4\eta
_{m}z}=\mp 2z(y-1).
\end{equation*}%
Multiplying $\sqrt{\left( z-1\right) ^{2}+4\eta _{A}z}-\sqrt{\left(
z-1\right) ^{2}+4\eta _{B}z}$ on both sides, we have%
\begin{equation*}
\sqrt{\left( z-1\right) ^{2}+4\eta _{A}z}-\sqrt{\left( z-1\right) ^{2}+4\eta
_{B}z}=\mp \frac{2(\eta _{A}-\eta _{B})}{y-1}.
\end{equation*}%
Adding the above two equations and taking the square, we further obtain%
\begin{equation*}
\left( z-1\right) ^{2}+4\eta _{A}z=\left( z(y-1)+\frac{\eta _{A}-\eta _{B}}{%
y-1}\right) ^{2}.
\end{equation*}%
Solving $z$, we obtain%
\begin{equation*}
S_{AB}\left( z\right) =-\frac{1-\eta _{A}-\eta _{B}\mp \sqrt{(1-\eta
_{A}-\eta _{B})^{2}+\left( 2z-z^{2}\right) \left( \left( \frac{\eta
_{A}-\eta _{B}}{z-1}\right) ^{2}-1\right) }}{2z-z^{2}}.
\end{equation*}%
From (2.45) in \cite{TulinoTextBook}, the a.e.d. of $\mathbf{U}_{A}\mathbf{U}%
_{A}^{\dagger }+\mathbf{U}_{B}\mathbf{U}_{B}^{\dagger }$ is given by%
\begin{equation*}
\mathcal{F}\left( \lambda \right) =\lim_{\omega \rightarrow 0^{+}}\frac{1}{%
\pi }\func{Im}\left[ S_{AB}\left( \lambda +j\omega \right) \right] .
\end{equation*}%
Thus, for $0<\lambda <1$ and $1<\lambda <2$, we obtain
\begin{equation}
\mathcal{F}\left( \lambda \right) =\frac{1}{\pi }\func{Im}\left[ \frac{\sqrt{%
(1-\eta _{A}-\eta _{B})^{2}+\left( 2\lambda -\lambda ^{2}\right) \left(
\left( \frac{\eta _{A}-\eta _{B}}{\lambda -1}\right) ^{2}-1\right) }}{%
2\lambda -\lambda ^{2}}\right] .  \label{F_lemda}
\end{equation}%
In addition, for a randomly generated pair of $\mathbf{U}_{A}$ and $%
\mathbf{U}_{B}$, there are $n_{A}+n_{B}-n_{R}$ orthogonal eigenvectors for $%
\lambda _{i}=2$, $\left\vert n_{A}-n_{B}\right\vert $ orthogonal
eigenvectors for $\lambda _{i}=1$, and $n_{R}-n_{A}-n_{B}$ orthogonal
eigenvectors for $\lambda _{i}=0$. Thus, as $n_{R}$ tends to infinity, the
PDF $\mathcal{F}\left( \lambda \right) $ at $\lambda =2$ is given by $\left[
\eta _{A}+\eta _{B}-1\right] ^{+}\delta \left( \lambda -2\right) $; that at $%
\lambda =1$ is given by $\left\vert \eta _{A}-\eta _{B}\right\vert \delta
\left( \lambda -1\right) $; and that at $\lambda =0$ is given by $\left[
1-\eta _{A}-\eta _{B}\right] ^{+}\delta \left( \lambda \right) $. This
concludes the proof of the lemma.

\section*{Acknowledgment}

The work of Xiaojun Yuan was partially supported by a grant from the
University Grants Committee (Project No. AoE/E-02/08) of the Hong Kong
Special Administrative Region, China. The work of Tao Yang was supported by
CSIRO OCE Postdoctoral Fellowships. It is also supported under the
Australian Government's Australian Space Research Program.

\newpage
\begin{figure}[tp]
\centering%
\includegraphics[width=4.2in,
height=1.8in]{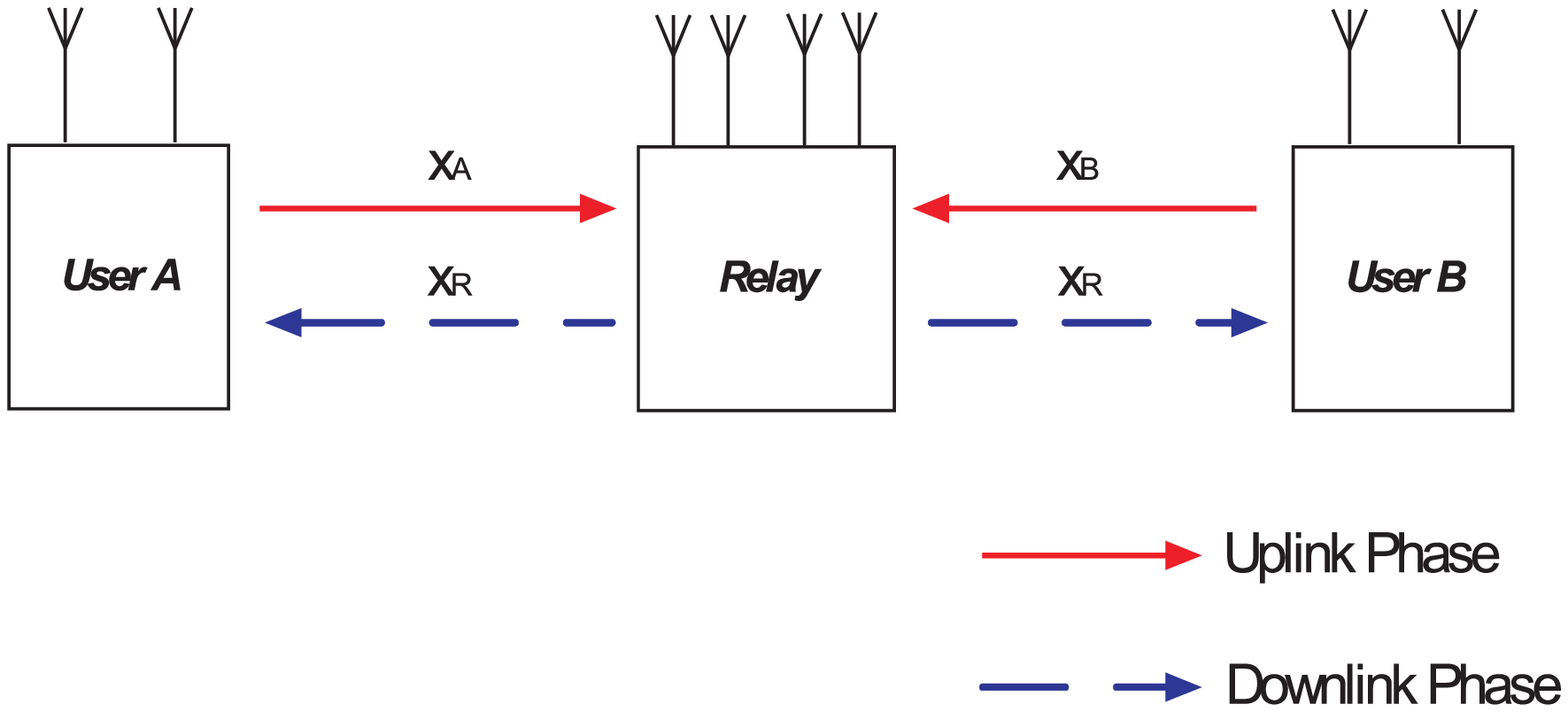}
\caption{Configuration of a MIMO TWRC.}
\label{Fig_Config_MIMOTWRC1}
\end{figure}
\begin{figure}[tp]
\centering\includegraphics[width=7.33in,
height=4.0in]{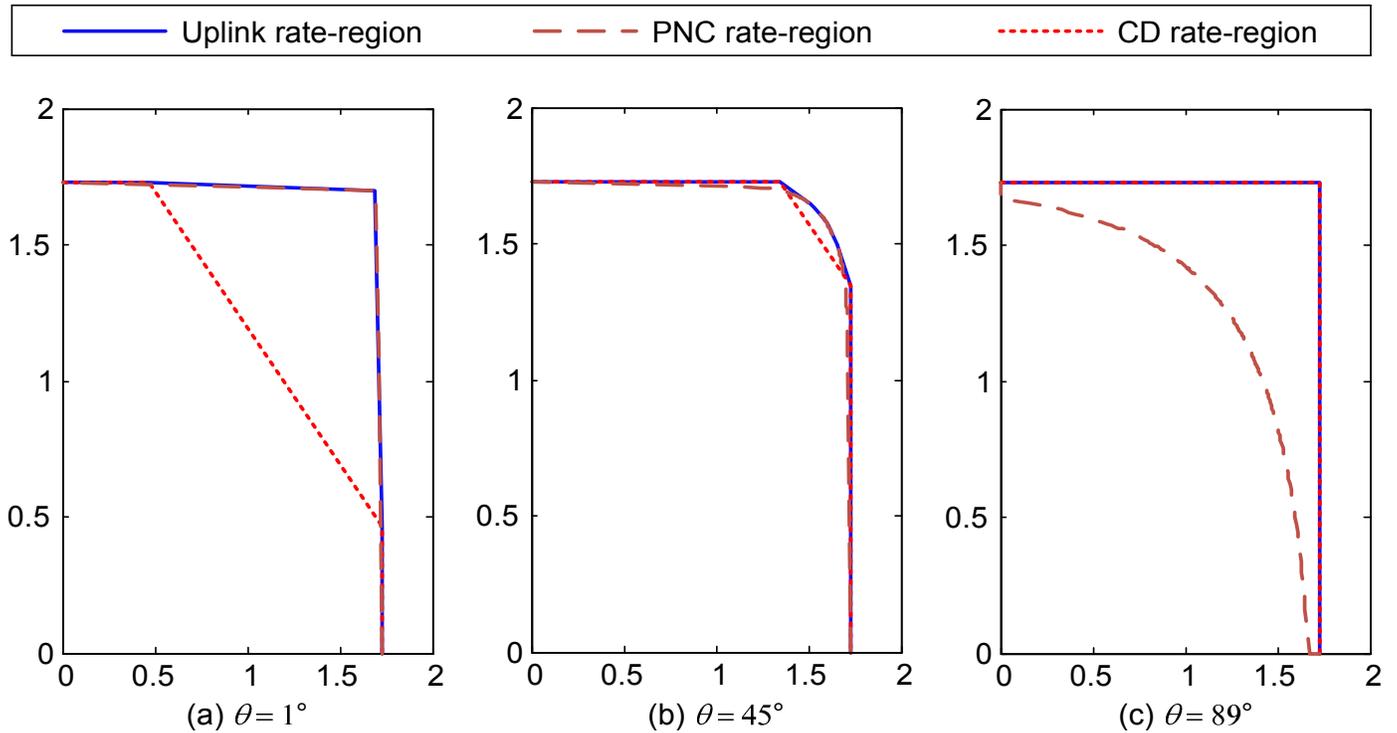}
\caption{The uplink rate-regions of the TWRCs with single-antenna users. $%
\mathbf{h}_{A}=[1,0]^{T}$ and $\mathbf{h}_{B}=[\cos \protect\theta ,\sin
\protect\theta ]^{T}$. Channel SNR = $1/N_{0}$ = 10 dB. The horizontal
axises represent the rate of user $A$; the vertical axises represent the
rate of user $B$; the unit is \textit{bit per channel use}.}
\label{SIMO}
\end{figure}
\begin{figure}[tp]
\centering\includegraphics[width=4.6in,height=3.8in]{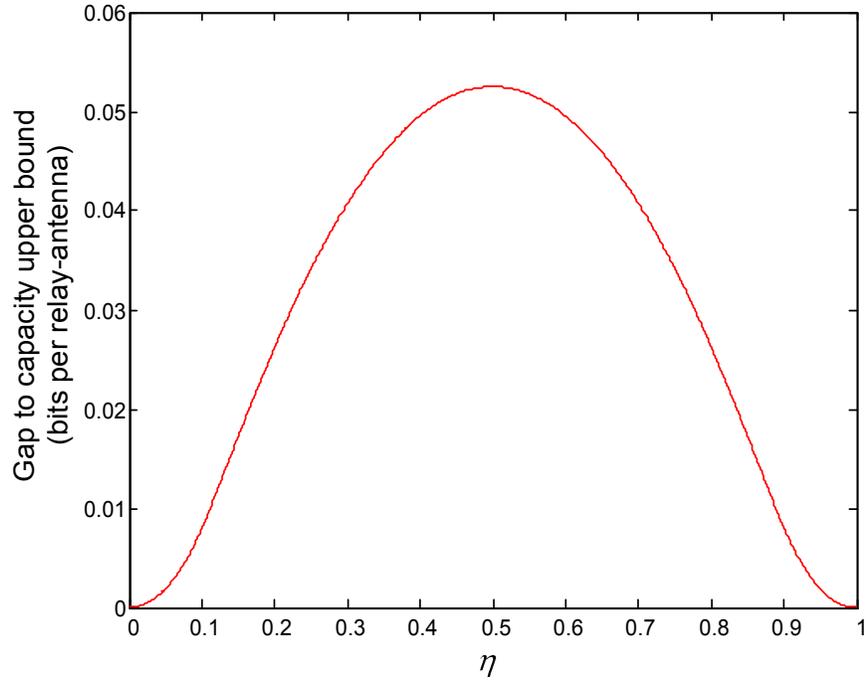}
\caption{The function of the average normalized gap $r^{SD}$ in (\ref{Limit1}) against $%
\protect\eta $.}
\label{SD_rate_gp}
\end{figure}
\begin{figure}[tp]
\centering\includegraphics[width=4.6in,height=3.8in]{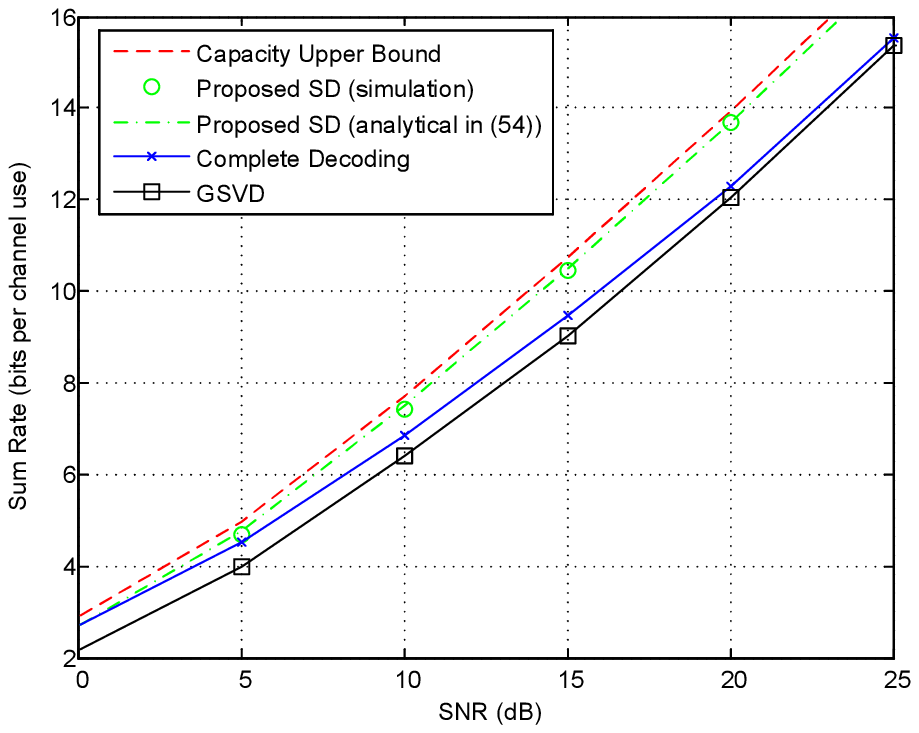}
\caption{Average achievable sum-rates of various schemes for the Rayleigh
fading MIMO TWRC with $n_{A}=n_{B}=2$ and $n_{R}=4$.}
\label{Figure nT2nR4}
\end{figure}
\begin{figure}[tp]
\centering\includegraphics[width=4.6in,height=3.8in]{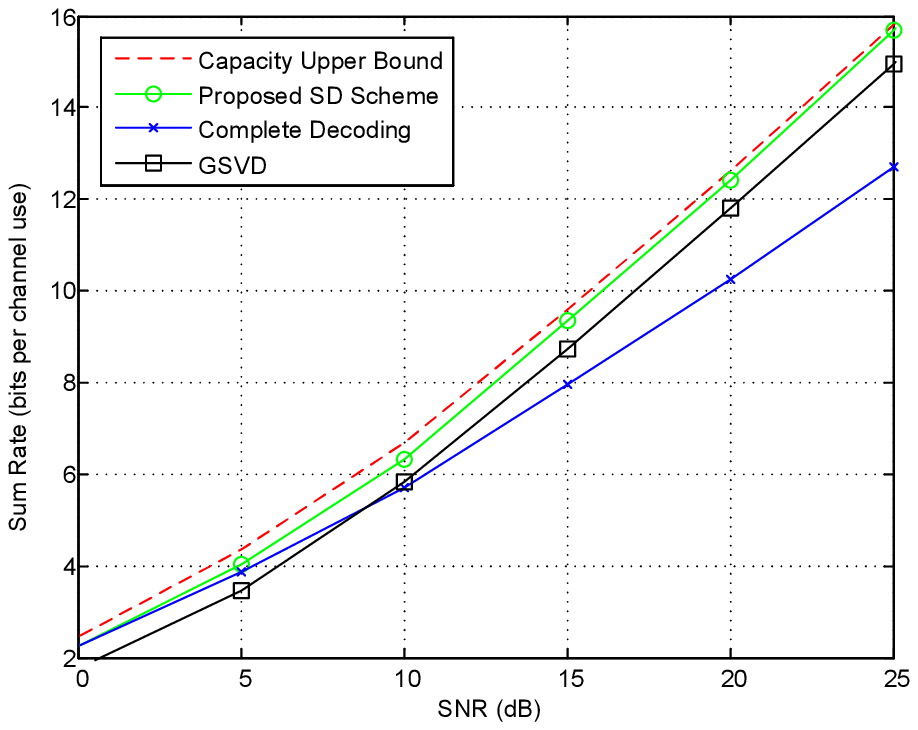}
\caption{Average achievable sum-rates of various schemes for the Rayleigh
fading MIMO TWRC with $n_{A}=n_{B}=2$ and $n_{R}=3$. }
\label{Figure nT2nR3}
\end{figure}
\begin{figure}[tp]
\centering\includegraphics[width=4.6in,height=3.8in]{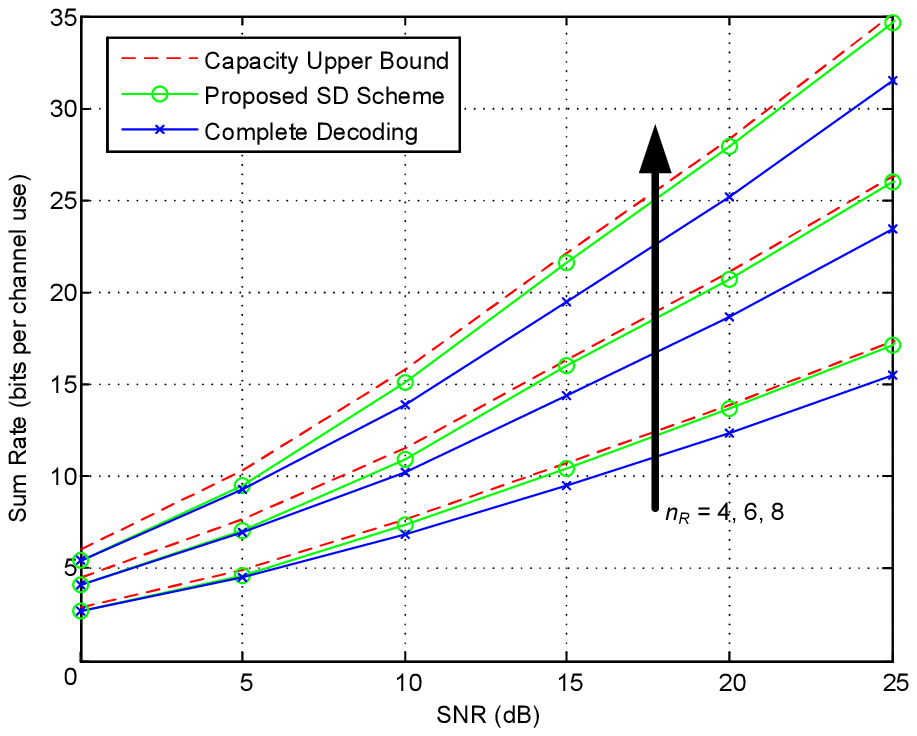}
\caption{Scaling effect of the average sum-rates of various schemes for the Rayleigh
fading MIMO TWRCs with $\protect\eta _{A}=\protect\eta _{B}=1/2$. }
\label{Figure Eda1_2}
\end{figure}
\begin{figure}[tp]
\centering\includegraphics[width=4.6in,height=3.8in]{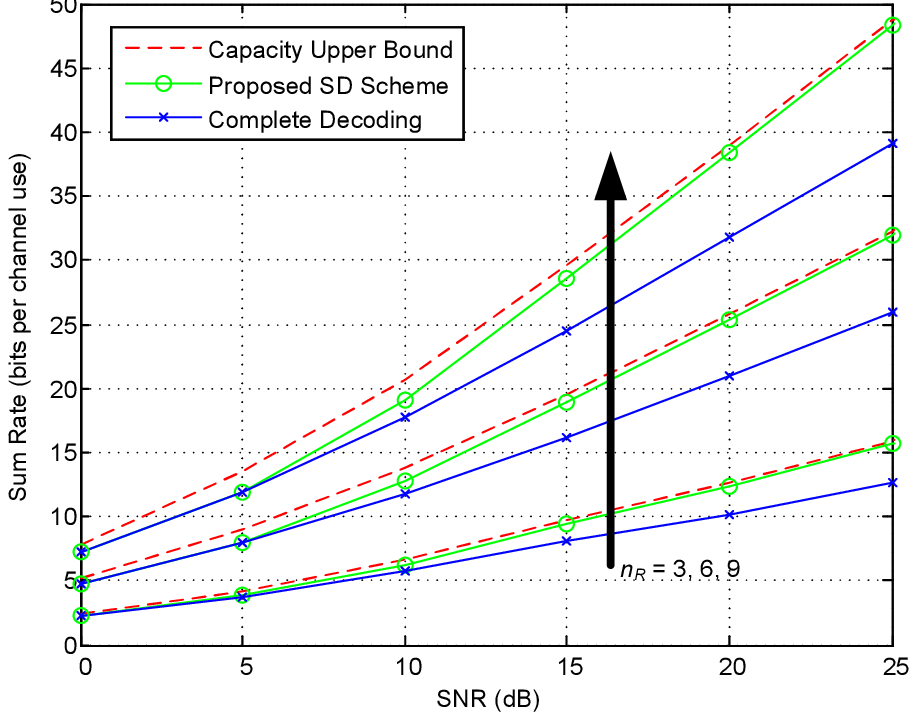}
\caption{Scaling effect of the average sum-rates of various schemes for the Rayleigh
fading MIMO TWRCs with $\protect\eta _{A}=\protect\eta _{B}=2/3$. }
\label{Figure Eda2_3}
\end{figure}
\begin{figure}[tp]
\centering\includegraphics[width=4.8in,height=4in]{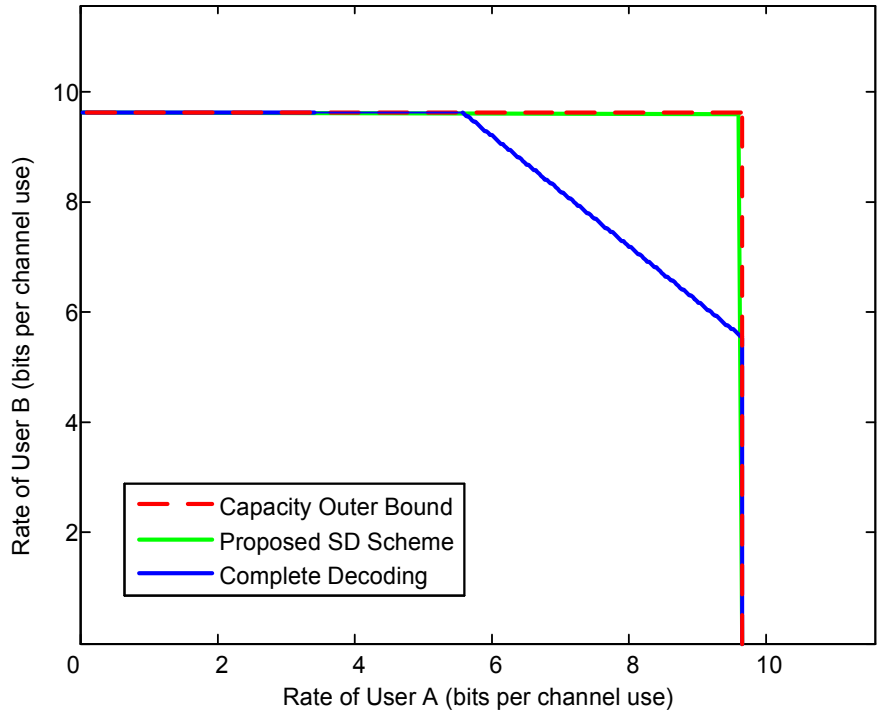}
\caption{Average achievable rate-regions for the Rayleigh
fading MIMO TWRC with $n_{A}=n_{B}=2$ and $n_{R}=3$. The average SNRs for all
the channel links are set to $30$ dB.}
\label{Figure CapacityRegion}
\end{figure}

\end{document}